\RequirePackage{afterpackage}
\AfterPackage{amsthm}{
  \RequirePackage{hyperref,cleveref}
}

\documentclass[a4paper,USenglish,thm-restate]{lipics-v2019}

\usepackage[normalem]{ulem}
\usepackage{float}
\usepackage{latexsym,amsmath,bm,setspace,afterpage,epsfig,graphicx,url,tabularx,xcolor,wrapfig,multirow,amsthm,amsfonts,openbib,multirow}
\usepackage{listings,color}
\usepackage[bottom]{footmisc}
\usepackage{xspace}
\usepackage{todonotes}
\usepackage{tikz}
\usepackage{enumitem}
\usepackage{epstopdf}
\usepackage[final]{pdfpages}
\usepackage{hyperref}
\usepackage{thmtools}
\usepackage{fancyhdr}
\usepackage{comment}
\usepackage{graphicx}
\usepackage[font={footnotesize,bf}]{caption}
\usepackage{url}
\usepackage{fancybox,algorithmic}
\usepackage{wrapfig}
\usepackage{array}
\newcolumntype{P}[1]{>{\centering\arraybackslash}p{#1}}
\usepackage{changepage}
\usepackage[sort,compress,noadjust]{cite}
\usepackage{thm-restate}
\usepackage{cleveref}
\usepackage{arydshln}

\crefname{section}{\S}{\S\S}

\crefname{appendix}{Appendix}{Appendices}

\renewcommand{\paragraph}[1]{\vspace{.05in} \noindent \textbf{#1}}

\crefname{theorem}{Theorem}{Theorems}
\crefname{definition}{Definition}{Definitions}
\crefname{proposition}{Proposition}{Propositions}
\crefname{lemma}{Lemma}{Lemmas}
\crefname{claim}{Claim}{Claims}
\crefname{corollary}{Corollary}{Corollaries}
\crefname{property}{Property}{Properties}
\crefname{algorithm}{Algorithm}{Algorithms}

\author{Daniel Katzan}{Tel Aviv University}{}{}{}
\author{Adam Morrison}{Tel Aviv University}{}{https://orcid.org/0000-0002-5586-2615}{}

\title{Recoverable, Abortable, and Adaptive Mutual Exclusion with Sublogarithmic RMR Complexity}
\titlerunning{Recoverable, Abortable, and Adaptive ME with Sublogarithmic RMR Complexity}
\authorrunning{D. Katzan and A. Morrison}
\Copyright{Daniel Katzan and Adam Morrison}

\funding{Israel Science Foundation (grant No. 2005/17) and Blavatnik ICRC at TAU}

\nolinenumbers

\newif\ifpaper
\papertrue

\newif\ifCR
\CRfalse

\ifCR
\else
\hideLIPIcs
\fi

\hypersetup{pdfstartpage=1,pdfpagemode=UseNone}
\usetikzlibrary{calc,arrows}

\newcommand{\Xomit}[1]{}

\definecolor{darkBlue}{rgb}{0,0,0.75}
\definecolor{darkGreen}{rgb}{0,0.5,0}
\ifCR
\relatedversion{The full version of the paper is available at~\url{http://arxiv.org/abs/2011.07622}.}
\EventEditors{Quentin Bramas, Rotem Oshman, and Paolo Romano}
\EventNoEds{3}
\EventLongTitle{24th International Conference on Principles of Distributed Systems (OPODIS 2020)}
\EventShortTitle{OPODIS 2020}
\EventAcronym{OPODIS}
\EventYear{2020}
\EventDate{December 14--16, 2020}
\EventLocation{Strasbourg, France  (Virtual Conference)}
\EventLogo{}
\SeriesVolume{184}
\ArticleNo{15}
\fi

\begin{document}
\maketitle

\ifpaper
\begin{abstract}
\else
\fi
We present the first recoverable mutual exclusion (RME) algorithm that is simultaneously abortable, adaptive to point contention, and
with sublogarithmic RMR complexity.  Our algorithm has $O(\min(K,\log_W N))$ RMR passage complexity and $O(F + \min(K,\log_W N))$ RMR super-passage complexity, where $K$ is the number of concurrent processes (point contention),
$W$ is the size (in bits) of registers, and $F$ is the number of crashes in a super-passage.
Under the standard assumption that $W=\Theta(\log N)$, these bounds translate to worst-case $O(\frac{\log N}{\log \log N})$ passage complexity and $O(F + \frac{\log N}{\log \log N})$ super-passage complexity.
Our key building blocks are:
\begin{itemize}[leftmargin=*]
\item A $D$-process abortable RME algorithm, for $D \leq W$, with $O(1)$ passage complexity and $O(1+F)$ super-passage complexity.  We obtain
this algorithm by using the Fetch-And-Add (FAA) primitive, unlike prior work on RME that uses Fetch-And-Store (FAS/SWAP).

\item A generic transformation that transforms any abortable RME algorithm with passage complexity of $B < W$,
into an abortable RME lock with passage complexity of $O(\min(K,B))$.
\end{itemize}   
\ifpaper
\end{abstract}
\else
\fi

\ifpaper
\section{Introduction}
\else
\fi

Mutual exclusion (ME)~\cite{dijkstra} is a central problem in distributed computing.
A mutual exclusion algorithm, or \emph{lock}, ensures that some \emph{critical section}
of code is accessed by at most one process at all times.  To enter the critical
section (CS), a process first executes an \emph{entry section} to \emph{acquire} the lock.
After leaving the critical section, the process executes an \emph{exit section} to
\emph{release} the lock.  The standard complexity measure for ME is
\emph{remote memory references} (RMR) complexity~\cite{Anderson:2002:ILB,RMR-bound}.
RMR complexity models the property that memory access cost on a shared-memory machine is
not uniform. Some accesses are \emph{local} and cheap, while the rest are \emph{remote}
and expensive (e.g., processor cache hits and misses, respectively).
The RMR complexity measure thus charges a process only for remote accesses.  There
are various RMR definitions, modeling cache-coherent (CC) and distributed shared-memory (DSM) systems.
The complexity of a ME algorithm is usually defined as its \emph{passage complexity}, i.e., the
number of RMRs incurred by a process as it goes through an entry and corresponding exit of
the critical section.

For decades, the vast majority of mutual exclusion algorithms were designed under the assumption that processes are \emph{reliable}:
they do not crash during the mutual exclusion algorithm or critical section. This assumption models
the fact that when a machine or program crashes, its memory state is wiped out. However, the recent
introduction of non-volatile main memory (NVRAM) technology can render this assumption invalid.
With NVRAM, memory state can remain persistent over a program or machine crash.  This change creates
the \emph{recoverable mutual exclusion} (RME) problem~\cite{first-recoverable}, of designing an
ME algorithm that can tolerate processes crashing and returning to execute the algorithm.
In RME, a \emph{passage} of a process $p$ is defined as the execution fragment from when $p$ enters
the lock algorithm and until either $p$ completes the exit section or crashes.
If $p$ crashes mid-passage and recovers, it re-enters the lock algorithm, which starts a new passage.
Such a sequence of $p$'s passages that ends with a crash-free passage (in which $p$ acquires
and releases the lock) is called a \emph{super-passage} of $p$.

RME constitutes an exciting clean slate for ME research. Over the years, locks with
many desired properties (e.g., fairness) were designed and associated complexity trade-offs were
explored~\cite{mutex-trends}.  These questions are now re-opened for RME, which has spurred a flurry
of research~\cite{sublog-first,sublog-second,journal,FCFS-RME,abortable-RME,first-recoverable,system-crash,optimal,adaptive-rme,amortized-rme}.
In this paper, we study such questions.  In a nutshell, we introduce an RME algorithm that is
\emph{abortable}, \emph{adaptive}, and has \emph{sublogarithmic} RMR complexity.
Our lock is the first RME algorithm adaptive to the number of concurrent processes (or \emph{point contention}) and the first \emph{abortable} RME algorithm with sublogarithmic RMR
complexity. It is also the first deterministic, worst-case sublogarithmic abortable lock in the DSM model (irrespective
of recoverability).  Our algorithm also features other desirable properties not present in prior work, as
detailed shortly.

\paragraph{Abortable ME \& RME}
An \emph{abortable} lock~\cite{AbortableMutex2,AbortableMutex1, amortized-abortable} allows a process waiting to acquire the lock
to give up and exit the lock algorithm in a finite number of its own steps.  Jayanti and Joshi~\cite{abortable-RME}
argue that abortability is even more important in the RME setting.  The reason is that a crashed process might
delay waiting processes for longer periods of time, which increases the motivation for allowing
processes to abort their lock acquisition attempt and proceed to perform other useful work.

Mutual exclusion, and therefore abortable ME (AME), incurs a worst-case RMR cost of $\Omega(\log N)$ in an $N$-process
system with standard read, write, and comparison primitives such as Compare-And-Swap (CAS) or LL/SC~\cite{RMR-bound}.
This logarithmic bound is achieved for both ME~\cite{Yang1995} and AME~\cite{adaptive-abortable}, and was recently achieved
for a recoverable, abortable lock by Jayanti and Joshi~\cite{abortable-RME}.  However, while there exists
an AME algorithm with sublogarithmic worst-case RMR complexity (in the CC model)~\cite{abortable-sublog},
no such abortable algorithm is known for RME.  Moreover, Jayanti and Joshi's $O(\log N)$ abortable RME algorithm is suboptimal in a few ways.
First, its worst-case RMR complexity is logarithmic only on a \emph{relaxed} CC model, in which a failed CAS on
a variable does not cause another process with a cached copy of the variable to incur an RMR on its next access to it,
which is not the case on real CC machines.
Their algorithm has linear RMR complexity in the realistic, standard CC model.
Second, their algorithm is starvation-free only if the number of aborts is finite.

\paragraph{Adaptive ME}
A lock is \emph{adaptive} with respect to \emph{point contention} if its RMR complexity depends on $K$, the number of processes concurrently trying to
access the lock, and not only on $N$, the number of processes in the system.  Adaptive locks are desirable because
they are often faster when $K \ll N$.  There exist locks with worst-case RMR cost of $O(\min(\log N, K))$ for
both ME~\cite{f-arrays} and AME~\cite{adaptive-abortable}, but no adaptive RME algorithm is known (independent of abortability).

\ifpaper
\subsection{Overview of Our Results}
\else
\section{Overview of Our Results}
\fi

In the following, we denote the number of crashes in a super-passage by $F$ and the size (in bits) of the system's
registers by $W$.  We obtain three keys results, which, when
combined, yield the first RME algorithm that is simultaneously abortable and adaptive, with worst-case $O(\log_W N)$
passage complexity and $O(F + \log_W N)$ super-passage complexity, in both CC and DSM models.  Assuming (as is standard)
that $W=\Theta(\log N)$, this translates to worst-case $O(\frac{\log N}{\log \log N})$ passage complexity and
$O(F + \frac{\log N}{\log \log N})$ super-passage complexity.
In contrast to Jayanti and Joshi's abortable RME algorithm~\cite{abortable-RME}, our lock
achieves sublogarithmic RMR complexity in the \emph{standard} CC model and is unconditionally
starvation-free.
Our algorithm's space complexity is a static (pre-allocated) $O(N W \log_W N)$ memory words
(which translates to $O(\frac{N \log^2 N}{\log \log N})$ if $W=\Theta(\log N)$). Jayanti and Joshi's
algorithm also uses static memory, but it relies on unbounded counters. The other sublogarithmic RME algorithms~\cite{sublog-first, sublog-second,adaptive-rme} use dynamic memory allocation, and may consume unbounded space.

\paragraph{Result \#1: $W$-process abortable RME with $O(1)$ passage and $O(1+F)$ super-passage complexity (\cref{chap:w_port})}
Our key building block is a $D$-process algorithm, for $D \leq W$.  It has constant RMR cost for a passage, regardless of if the
process arrives after a crash.  The novelty of our algorithm is that it uses the Fetch-And-Add (FAA) primitive to beat
the $\Omega(\log D)$ passage complexity lower-bound. In contrast, the building blocks in prior RME work with worst-case
sublogarithmic RMR complexity use the Fetch-And-Store
(FAS, or SWAP) primitive and assume no bound on $D$, even though they are ultimately used by only a bounded
number of processes in the final algorithm.  By departing from FAS and exploiting the process usage bound, we overcome
difficulties that made the prior algorithms' building blocks~\cite{sublog-first,sublog-second} have only $O(D)$ RMR passage
complexity.

These prior algorithms use a FAS-based queue-based lock as a building block.  They start with an $O(1)$ RMR queue-based ME
algorithm~\cite{MCS,CLH}, in which a process trying to acquire the lock uses FAS to append a node to the queue tail, and then
spins on that node waiting for its turn to enter the critical section.  Unfortunately, if the process crashes after the FAS,
before writing its result to memory, then when it recovers and returns to the algorithm, it does not know whether it
has added itself to the queue and/or who is its predecessor (previously obtained from the FAS response).  To overcome this
problem, a recovering process reconstructs the queue state into some valid state, \emph{which incurs a linear number of RMRs}.
The recovery procedure is blocking (not wait-free), and multiple processes cannot recover concurrently.  Overall, these prior
building blocks have $O(D)$ passage complexity and $O(1+FD)$ super-passage complexity.  In contrast, our $D$-process abortable
RME algorithm has $O(1)$ passage complexity and $O(1+F)$ super-passage complexity, has wait-free recovery, and allows multiple
processes to recover concurrently.  While other $O(1)$ RME algorithms exist, they either assume a weaker crash
model~\cite{system-crash}, rely on non-standard primitives that are not available on real machines~\cite{optimal,sublog-first},
or obtain only  amortized, not worst-case, $O(1)$ RMR complexity~\cite{amortized-rme}.

\paragraph{Result \#2: Tournament tree with wait-free exit (\cref{chap:tree})}
In both ours and prior work~\cite{sublog-first,sublog-second}, the main lock is obtained by constructing a \emph{tournament tree}
from the $D$-process locks.  The tree has $N$ leaves, one for each process. Each internal node is a $D$-process lock, so
the tree has height $O(\log_D N)$.
To acquire the main lock, a process competes to acquire each lock on the path from its leaf to the root, until it wins
at the root and enters the critical section.  Our algorithm differs from prior tournament trees in a couple of simple ways, but which have important
impact.

First: In our tree, a process that recovers from a crash returns directly to the node in which it crashed.  This allows us to
leverage our node lock's $O(1+F)$ super-passage complexity to obtain $O(H + F)$ super-passage complexity for the tree, where $H$
is the tree's height.  By taking $D=W=\Theta(\log N)$, our overall lock has $O(F + \frac{\log N}{\log \log N})$ super-passage
complexity and $O(\frac{\log N}{\log \log N})$ passage complexity. In contrast, prior trees perform recovery by having a process
restart its ascent from the leaf. In fact, in these algorithms, there is no asymptotic benefit from returning directly to the
node where the crash occurred. The reason is that node lock recovery in these trees has $O(D)$ complexity, so to obtain overall
sublogarithmic complexity, they take $D=\frac{\log N}{\log \log N}$, which means that node crash recovery costs the
same as climbing to the node.  Consequently, their overall super-passage complexity is multiplicative in $F$,
$O((1+F)\frac{\log N}{\log \log N})$, instead of additive as in our tree.

Second: Our tree's exit section is wait-free (assuming finitely many crashes). In contrast, in the prior trees, a process that
crashes during its exit section might subsequently block.  The reason is a subtle issue related to composition of RME locks.
The model in these works~\cite{sublog-first,sublog-second} is that a process $p$ that crashes in its exit section must complete a crash-free
passage upon recovery (i.e., re-enter the critical section and exit it again). Thus, $p$ must re-ascend to the root after recovering.
Each node lock satisfies a \emph{bounded CS re-entry} property, which allows $p$ to re-enter the node's CS (i.e., ascend)
without blocking---provided that $p$ crashed inside the node's CS.  However, this property does not apply if $p$ released
the node lock (i.e., descended) before crashing. For such a node, $p$ simply attempts to re-acquire the node lock.
Consequently, $p$ might block during its recovery, even though logically \emph{it is only trying to release the overall lock}.
We address this problem by carefully modeling the interface of an RME algorithm in a way that facilitates composition,
which enables a recovering process to avoid re-acquiring node locks it had already released.  Our overall algorithm thereby
satisfies a new \emph{super-passage wait-free exit} property.

\paragraph{Result \#3: Generic RME adaptivity transformation (\cref{chap:transform})}
We present a generic transformation that transforms any abortable RME algorithm with passage complexity of $B < W$ into an abortable RME lock with passage complexity of $O(\min(K,B))$, where $K$ is the number of processes executing the algorithm concurrently with the process going
through the super-passage, i.e., the \emph{point contention}.
Applying this transformation to our tournament tree lock yields the final algorithm.

\ifpaper

\else
\section{Related Work}
The RME problem was first formulated and introduced by Golab and Ramaraju in~\cite{first-recoverable} \cite{journal}, where they presented the first RME algorithm using CAS and FAS primitives with worst-case RMR complexity of $O(N)$. Their algorithm is based on a building block of an RME algorithm designed for 2 processes accessing it simultaneously, and then compose those building blocks to a tournament tree, where they eventually bound the worst-case RMR of their algorithm by using a fast-path and slow-path mechanism. In~\cite{FCFS-RME} Jayanti and Joshi presented the first FCFS RME algorithm with logarithmic RMR complexity, this is based on their min-array object which they have made recoverable, and later also added it the abortable property in~\cite{abortable-RME}. The logarithmic bound was first beaten by~\cite{sublog-first} for the CC model, and later in~\cite{sublog-second} for both the CC and DSM model as well as fixing subtle bugs encountered in~\cite{sublog-first}, both papers do so by using a building block of an abortable RME algorithm designed for $D$ processes, and then composing them into a tournament tree, taking $D=O(\frac{\log N}{\log \log N})$ yields a sub-logarithmic RMR complexity. Optimal (constant worst-case complexity) RME algorithms were introduced in~\cite{optimal} and~\cite{system-crash}, but the former assumed a strong FASAS primitive which does not exist in common processors, and the latter a weaker crash-recovery model where they consider a system-crash where all processes crash simultaneously instead of each process crashing independently. An adaptive RME algorithm with respect to the number of crashes in the system, i.e. an RME algorithm whose RMR complexity depends on $F^*$ the number of overall crashes in the system and not only on $N$ the number of processes in the system, was recently introduced in~\cite{adaptive-rme} which has a better RMR complexity in the case where the number of crashes is small. A constant amortized version of RME algorithm was presented in~\cite{amortized-rme} by using a FAA instead of FAS. To achieve that, they based their RME algorithm on a ticketing mechanism, where each process acquires a ticket, and then announces its ticket number on its entry section. A process releasing the lock hands off the lock to the next process in line that has already announced its ticket number. It might be the case that a ticket number was acquired by a process, but it was too slow to announce it, this process will be skipped and the lock will be handed off to a process holding a higher ticket number, causing the skipped process to retry its entry section and acquire a new ticket number. The fact that this can only happen to a process once for every concurrent passage is used to achieve the amortized cost.
\fi

\begin{table}
\begin{adjustwidth}{-0.6cm}{}
\resizebox{1.1\textwidth}{!}{%
\begin{tabular}{ |P{3.1cm}|P{3.9cm}|P{3.9cm}|P{1.6cm}|P{1.8cm}|P{3.8cm}| }

\hline

\multirow{2}{*}{\textbf{Algorithm}}& \textbf{Passage} & \textbf{Super-Passage} & \textbf{Primitives} & \textbf{Space} & \textbf{Additional} \\
& \textbf{Complexity} & \textbf{Complexity} & \textbf{Used} & \textbf{Complexity} & \textbf{Properties} \\

 \hline

 \multirow{3}{*}{\shortstack{Golab \&  Ramaraju \\ \cite[Section 4.2]{journal} \\with MCS~\cite{MCS} as \\base lock}} & \multicolumn{1}{l|}{$O(1)$ (no concurrent crashes)} & $O(1)$ (no concurrent crashes) & \multirow{3}{*}{CAS, FAS} & \multirow{3}{*}{$O(N  \log N)$} & \multirow{3}{*}{ } \\ \cdashline{2-3}
& $O(\log N)$ (concurrent crashes) & $O(\log N)$ (concurrent crashes) &&& \\ \cdashline{2-3}
 & $O(N)$ (if crashes) & $O( F  N)$ (if crashes) &&& \\

 \hline

 Jayanti \& Joshi~\cite{FCFS-RME} & $O(\log N)$ & $O(\log N + F)$ & CAS & $O(N  \log N)$ &  FCFS, SP WF Exit\\

 \hline

Jayanti, Jayanti, \& Joshi~\cite{sublog-second} & $O(\frac{\log N}{\log \log N})$ & $O((1 + F)  (\frac{(\log N}{\log \log N}))$ & FAS & Unbounded &  \\

 \hline

Jayanti, Jayanti, \& Joshi~\cite{optimal} & $O(1)$ & $O(1)$ in the DSM model \newline $O(F)$ in the CC model & FASAS & $O(N)$ &  SP WF Exit\\

 \hline

Chan \& Woelfel~\cite{amortized-rme} &  \Xomit{$O(1 + F^*)$ worst-case \newline} $O(1)$ \emph{amortized} & same as passage complexity & CAS, FAA & Unbounded &  SP WF Exit\\

  \hline

Dhoked \& Mittal~\cite{adaptive-rme} & $O(min(\sqrt {F^*} , \frac{\log N}{\log \log N})$ & same as passage complexity & CAS, FAS & Unbounded & \textbf{Crash-adaptive} \\

\hline

Jayanti \& Joshi~\cite{abortable-RME} & $O(\log N)$ & $O(\log N + F)$ & CAS & $O(N  \log N)$ &  \textbf{Abortable}, SP WF Exit, FCFS\\

 \hline\hline

\textbf{ This work} & $\pmb{O(\min(K, \frac{\log N}{\log \log N}))}$ & $\pmb{O(\min(K, \frac{\log N}{\log \log N}) + F)}$ & \textbf{FAA, CAS} & $\pmb{O(\frac{N \log^2 N }{\log \log N})}$ &  \textbf{Abortable, adaptive, SP WF Exit} \\

 \hline

\end{tabular}%
}
\ifpaper
\begin{minipage}{1.1\textwidth}
\fi
\caption{\label{compare} Comparison of RME algorithms. (SP: super-passage, WF: wait-free, $\mathbf{F^*}$: total number of crashes in the system, FASAS: Fetch-And-Swap-And-Swap.)
All algorithms satisfy starvation-freedom, wait-free critical-section re-entry, and wait-free exit (defined in~\cref{chap:Model}).
}
\ifpaper
\end{minipage}
\fi
\end{adjustwidth}
\vspace{-20pt}
\end{table}

\paragraph{Summary of contributions and related work}
Table~\ref{compare} compares our final algorithm to prior RME work.
Dhoked and Mittal~\cite{adaptive-rme} use a definition of ``adaptivity'' that requires RMR cost to depend on the total number of crashes;
we refer to this property as \emph{crash-adaptivity}.
Crash-adaptivity is thus orthogonal to the traditional notion of adaptivity~\cite{AdaptiveDef}.
Chan \& Woelfel's algorithm~\cite{amortized-rme} uses FAA, but it is used to assign processes with tickets, which is different from
our technique (\Cref{chap:w_port}). Their algorithm has only an amortized RMR passage complexity bound and its worst-case RMR cost is unbounded.

\ifpaper
\else
All algorithms in the table satisfy the desirable RME properties of starvation-freedom, critical-section re-entry, wait-free critical-section re-entry, and wait-free exit conditions (see~\cref{chap:Model}).
\fi

\ifpaper
\section{Model and Preliminaries} \label{chap:Model}
\else
\fi

\paragraph{Model}
We consider a system in which $N$ deterministic, asynchronous, and unreliable processes communicate over a shared
memory. The shared memory, $M$, is an array of $\Theta(W)$-bit words. (Henceforth, we refer to the shared memory simply as ``memory'';
process-private variables are not part of the shared memory.)
The system supports the standard read, write, CAS, and FAA operations.
$CAS(a, o, n)$ atomically changes $M[a]$ from $o$ to $n$ if $M[a]=o$ and returns true; otherwise, it returns false without changing
$M[a]$.  $FAA(a,x)$ atomically adds $x$ to $M[a]$ and returns $M[a]$'s original content.

A \emph{configuration} consists of the state of the memory and of all processes, where the state of process $p$ consists of its
internal program counter and (non-shared) variables.  Given a configuration $\sigma$, an \emph{execution fragment} is a (possibly infinite)
sequence of \emph{steps}, each of which moves the system from one configuration to another, starting from $\sigma$.
In a \emph{normal step}, some process $p$ invokes an operation on a memory word and receives the operation's response.
In a \emph{crash} step, the state of some process $p$ resets to its initial state (but the memory state remains unchanged).
An \emph{execution} is an execution fragment starting from the system's initial configuration.

\paragraph{Notation}
Given an execution fragment $\alpha$, if $\beta$ is
a subsequence of $\alpha$, we write $\beta \subseteq \alpha$. If $e$ is a step taken in $\alpha$, we write $e \in \alpha$.
If $e$ is the $t$-th step in an execution $E$, we say that \emph{$e$ is at time $t$}.  We use $[t,t']$ to denote the subsequence
of $E$ whose first and last steps are at times $t$ and $t'$ in $E$, respectively.

\paragraph{RMR complexity}
The RMR complexity measure breaks the memory accesses by a process $p$ into \emph{local} and \emph{remote} references,
and charges $p$ only for remote references.  We consider two types of RMR models.  In the DSM model, each memory word is
local to one process and remote to all others, and process $p$ performs an RMR if it accesses a memory word remote to it.
In the CC model, the processes are thought of as having coherent caches, with RMRs occurring when a process accesses an
uncached memory word. Formally: (1) every write, CAS, or FAA is an RMR, and (2) a read by $p$ of word $x$ is an RMR if it is the first
time $p$ accesses $x$ or if after $p$'s prior access to $x$, another process performed a write, CAS, or FAA on $x$.

\paragraph{Recoverable mutual exclusion (RME)}
Our RME model draws from the models of Golab and Ramaraju~\cite{journal} and Jayanti and Joshi~\cite{FCFS-RME}.
In the spirit of~\cite{journal}, we model the RME algorithm as an object exporting methods invoked
by a client process. 
In the spirit of~\cite{FCFS-RME},
we require recovery to re-execute the section in which the crash occurred, rather than restart the entire passage.
An \emph{RME algorithm} (or \emph{lock}) provides the methods
\emph{Recover}, \emph{Try}, and \emph{Exit}.
(In the code, we show the methods taking an argument specifying the calling process' id.)
If process $p$ invokes \emph{Try} and it returns TRUE, then $p$ has \emph{acquired}
the lock and \emph{enters} the critical section (CS). Subsequently, $p$ \emph{exits} the CS by invoking \emph{Exit}.  If \emph{Exit}
completes, we say that $p$ has \emph{released} the lock.
The \emph{Recover} method guides $p$'s execution after a crash, which resets $p$ to its initial state.
We assume $p$'s initial state is to invoke \emph{Recover}, which returns $r \in \{TRY,CS,EXIT\}$.
If $r=TRY$, $p$ invokes \emph{Try}.  If $r=CS$, $p$ enters the CS. If $r=EXIT$, $p$ invokes \emph{Exit}.

A \emph{super-passage} of $p$ begins with $p$ completing \emph{Recover} and invoking \emph{Try}, either for the first time,
or for the first time after $p$'s prior super-passage ended.
The super-passage ends when $p$ completes \emph{Exit}. A \emph{passage} of $p$ begins with $p$ starting
a super-passage, or when $p$ invokes \emph{Recover} following a crash step. The passage ends at the earliest of $p$
completing \emph{Exit} or crashing.  We refer to an $L$-passage (or $L$-super-passage) to denote the lock
$L$ that a passage (or super-passage) applies to; similarly, we refer to a step taken in lock $L$'s code as an $L$-step.
We omit $L$ when the context is clear.
These definitions facilitate composition of RME locks. For instance, suppose that process $p$ is releasing locks
in a tournament tree and crashes after releasing some node lock $L$.  When $p$ recovers, it can invoke \emph{L.Recover},
which will return $TRY$, and thereby learn that it has released $L$ and can descend from it---without the \emph{Recover}
invocation counting as starting a new $L$-super-passage.

\emph{Well-formed} executions formalize the above described process behavior:

\begin{definition} \label{defn:well-formed}
An execution is \emph{well-formed} if the following hold for every lock $L$ and process $p$:
\begin{enumerate}
\item \label{prop:recover} \emph{Recover invocation}: $p$'s first $L$-step after a crash step is to invoke $L.Recover$.
\item \label{prop:try} \emph{Try invocation}: $p$ invokes $L.Try$ only if $p$ is starting a new $L$-super-passage, or if $p$'s prior crash step was during $L.Try$.
\item \label{prop:cs} \emph{CS invocation}: $p$ enters the CS of $L$ only if $p$ receives $TRUE$ from $L.Try$ in its current $L$-passage, or if $p$'s prior crash step was during the CS.
\item \label{prop:exit} \emph{Exit invocation}: $p$ invokes $L.Exit$ only if $p$ is in the CS of $L$, or if $p$'s prior crash step was during $L.Exit$.
\end{enumerate}
\end{definition}

Henceforth, we consider only well-formed execution. We also consider only \emph{well-behaved} RME algorithms, in which \emph{Recover} correctly identifies where a process crashes:

\begin{definition}
An RME algorithm is \emph{well-behaved} if the following hold, for every process $p$ and every \emph{well-formed} execution:
\begin{enumerate}
\item $p$'s first complete invocation of \emph{Recover}, and $p$'s first complete invocation of \emph{Recover} following a complete passage of \emph{Exit}, returns $TRY$.
\item $p$'s first complete invocation of \emph{Recover} following a crash during \emph{Try} return $TRY$.
\item $p$'s first complete invocation of \emph{Recover} following a crash during the CS returns $CS$.
\item $p$'s first complete invocation of \emph{Recover} following a crash during \emph{Exit} returns $EXIT$.
\item A complete invocation of \emph{Recover} by $p$ during the CS returns $CS$.
\end{enumerate}
\end{definition}

\noindent
Note: We consider $p$ to be in the \emph{Try} or \emph{Exit} section from the time it executes the first memory operation
of that section and until it either crashes or executes the last memory operation of that section. Thus, $p$ is considered
to be in the CS after it executes its final \emph{Try} memory operation.

\paragraph{Fairness}
We make a standard fairness assumption on executions: once $p$ starts a super-passage, it does not stop taking steps until
the super-passage ends.

\paragraph{Abortable RME}
At any point during its super-passage, process $p$ can non-deterministically choose to abort its attempt, which we model
by $p$ receiving an external \emph{abort} signal that remains
visible to $p$ throughout the super-passage (i.e., including after crashes) and resets once $p$ finishes the super-passages.
Abortable RME extends the definition of a super-passage as follows.  If $p$ is signalled to abort and its execution of
\emph{Try} returns FALSE, then $p$ has aborted and the super-passage ends.  (It is not mandatory for \emph{Try} to return
FALSE, because an abort may be signalled just as $p$ acquires the lock.)

\paragraph{$D$-ported locks}
    We model locks that may be used by at most $D$ processes concurrently as follows.  In a \emph{$D$-ported lock}, each
    process invokes the methods with a \emph{port} argument, $1 \leq k \leq D$, which acts as an identifier.  We augment the
definition of a well-formed execution to include the following conditions:
\begin{enumerate}
\setcounter{enumi}{4}
\item \label{prop:port} \emph{Constant port usage}: For every process $p$ and $L$-super-passage of $p$, $p$ does not change its port for $L$ throughout the super-passage.
\item \label{prop:sp}\emph{No concurrent super-passages}: For any $L$-super-passages $sp_i$ and $sp_j$ of processes $p_i \ne p_j$, if $sp_i$ and $sp_j$ are concurrent,
 then $p_i$'s port for $L$ in $sp_i$ is different than $p_j$'s port for $L$ in $sp_j$.  (Two super-passages are not \emph{concurrent}
 if one ends before the other begins.)
\end{enumerate}

\paragraph{Problem statement}
Design a well-behaved abortable RME algorithm with the following properties.

\begin{enumerate}
\item \textbf{Mutual exclusion}: At most one process is in the CS at any time $t$.
\item \textbf{Deadlock-freedom}: If a process $p$ starts a super-passage $sp$ at time $t$, and does not abort $sp$, and if every process that enters the CS eventually leaves it, then there is some time $t'>t$ and some process $q$ such the $q$ enters the $CS$ in time $t'$,
    \emph{or else there are infinitely many crash steps}.
\item \textbf{Bounded abort}: If a process $p$ has abort signalled while executing \emph{Try}, and executes sufficiently many steps without crashing, then $p$ complete its execution of \emph{Try}.
\end{enumerate}

The following properties are also desirable, and all but FCFS are satisfied by our algorithm:

\begin{enumerate}
\setcounter{enumi}{3}
\item \textbf{Starvation-freedom}: If the total number of crashes in the execution is finite and process $p$ executes infinitely many steps and every process that enters the CS eventually leaves it, then $p$ enters the CS in each super-passage in which it does not receive an abort signal.
\item \textbf{CS re-entry}: If process $p$ crashes while in the CS, then no other process enters the CS from the time $p$ crashes to the time when $p$ next enters the CS.
\item \textbf{Wait-free CS re-entry}: If process $p$ crashes in the $CS$, and executes sufficiently many steps without crashing, then $p$ enters the CS.
\item \textbf{Wait-free exit}: If process $p$ is executing \emph{Exit}, and executes sufficiently many steps without crashing, then $p$ completes its execution of \emph{Exit}.
\item \textbf{Super-passage wait-free exit}: If process $p$ is executing \emph{Exit}, then $p$ completes an execution of \emph{Exit} after a finite number of its own steps, \emph{or else $p$ crashes infinitely many times}.  (Notice that $p$ may crash and return to re-execute \emph{Exit}.)
\item \textbf{First-Come-First-Served (FCFS)}: If there exists a bounded section of code in the start of the entry section, referred to as the \emph{doorway} such that, if process $p_i$ finishes the doorway in its super-passage $sp_i$ for the first time before some process $p_j$ begins its doorway for the first time in its super-passage $sp_j$, and $p_i$ does not abort $sp_i$, then $p_j$ does not enter the CS in $sp_j$ before $p_i$ enters the CS in $sp_i$.
\end{enumerate}

Super-passage wait-free exit is a novel property introduced in this work.  It guarantees that a process
completes \emph{Exit} in a finite number of its own steps, as long as it only crashes finitely many times. Wait-free exit does not imply super-passage wait-free exit since it does not apply if the process crashes during \emph{Exit}.
Clearly, starvation-freedom implies deadlock-freedom, wait-free CS re-entry implies CS re-entry, and super-passage wait-free exit implies wait-free exit.

\paragraph{Lock complexity}
The \emph{passage complexity} (respectively, \emph{super-passage complexity}) of a lock is the maximum number of RMRs that a process
can incur while executing a passage (respectively, super-passage).  We denote by $F$ the maximum number of times a process crashes in
an execution.

\ifpaper
\section{$W$-Port Abortable RME Algorithm} \label{chap:w_port}
\else
\fi

Here, we present our $D$-process abortable RME algorithm, for $D \leq W$, which has $O(1)$ passage RMR complexity and $O(1+F)$ super-passage complexity.
The algorithm is similar in structure to Jayanti and Joshi's abortable RME algorithm~\cite{abortable-RME}, in that
it is built around a recoverable auxiliary object that tracks the processes waiting to acquire the lock.  This
object's RMR complexity determines the algorithm's complexity.  Non-abortable RME locks implement such an object
with a FAS-based linked list~\cite{sublog-first, sublog-second}.  Such a list has $O(1+FD)$ super-passage complexity---i.e.,
a crash-free passage incurs $O(1)$ RMRs---but it is hard to make abortable.  Jayanti and Joshi instead use a
recoverable \emph{min-array}~\cite{f-arrays}. This object supports aborting, but its passage complexity is logarithmic,
even in the absence of crashes.

Our key idea is to represent the ``waiting room'' object with a FAA-based $W$-bit mask (a single word), where
a process $p$ arriving/leaving is indicated by flipping a bit associated with $p$'s port.  The key ideas are that (1)
if $p$ crashes and recovers, it can learn its state in $O(1)$ RMRs simply by reading the bit mask and (2) the algorithm
carefully avoids relying on any FAA's return value. Our design thus
obtains the best of both worlds: the object can be updated with $O(1)$ RMRs as well as supports efficient aborting
(with a single bit flip).  The trade-off we make in this design choice is that we only guarantee starvation-freedom,
but not FCFS.  Unlike a min-array, the bit mask cannot track the order of arriving processes, as bit
setting operations commute.  We do, however, track the order in which processes acquire the lock, and thereby guarantee
starvation-freedom.

Our algorithm guarantees starvation-freedom unconditionally, even if there are infinitely many aborts.  This turns out
to be a subtle issue to handle correctly (\cref{abortable not wait-free}), and the Jayanti and Joshi algorithm is
prone to executions in which a process that does not abort starves as a result of other processes aborting infinitely
often (we show an example in~\cref{abortable not wait-free}).

\begin{figure}[!b]
\noindent
\hspace{-6ex}
\begin{minipage}{.58\textwidth}
\begin{lstlisting}[firstnumber=last]
 @{\bf ACTIVE}@: int // initially 0
 @{\bf STATUS}@: array of $W$ status words // initially all TRY
 @{\bf GO}@: array of $W$ pointers to booleans // initially all $\perp$
 @{\bf LOCK\_STATUS}@: struct {bool, port_id, bool*}
 // initially (0, 0, $\perp$)

void @{\bf Try}@(int k) {
    if STATUS[k] = ABORT:   @\label{rme:checkStatus}@
        Exit(k, TRUE)
        return FALSE
    if GO[k] = $\perp$:     @\label{rme:checkSpin}@
        if got abort signal: @\label{rme:TryAbort}@
            STATUS[k] := ABORT
            Exit(k, TRUE)
            return FALSE @\label{rme:TryAbortEnd}@
        GO[k] := new Bool() @\label{rme:checkSpinEnd}@
    if k-th bit in ACTIVE is 0: @\label{rme:Flip}@
        FAA(ACTIVE, $2^k$) @\label{rme:FlipEnd}@
    Promote($\perp$) @\label{rme:TryPromote}@
    while *GO[k] = FALSE:  @\label{rme:Spin}@
        if got abort signal:
            STATUS[k] := ABORT
            Exit(k, TRUE)
            return FALSE
@\label{rme:SpinEnd}@
    STATUS[k] := CS
    return TRUE
}
status @{\bf Recover}@(int k) {
    if STATUS[k] = EXIT:
        return EXIT
    if STATUS[k] = CS:
        return CS
    return TRY
}
\end{lstlisting}
\end{minipage}\hfill
\hspace{1ex}\begin{minipage}{.56\textwidth}
\begin{lstlisting}[firstnumber=last]
void @{\bf Exit}@(int k, bool abort) {
    if abort = FALSE: @\label{rme:ExitStatus}@
        STATUS[k] = EXIT @\label{rme:ExitStatusEnd}@
    if k-th bit in ACTIVE is 1:  @\label{rme:ExitClear}@
        FAA(ACTIVE, $-2^k$) @\label{rme:ExitClearEnd}@
    Promote(k) @\label{rme:ExitPromote1}@
    (taken, owner, owner_go) := LOCK_STATUS @\label{rme:ExitCheckOwn}@
    if taken = 1 and owner = k:
        CAS(LOCK_STATUS, (1, owner, owner_go), 
                (0, owner, owner_go)) @\label{rme:ExitCheckOwnEnd}@
    Promote($\perp$) @\label{rme:ExitPromote2}@
    if GO[k] $\neq \perp$ @\label{rme:ExitRetire}@
        Retire(GO[k])
        GO[k] := $\perp$  
    STATUS[k] := TRY @\label{rme:ExitRetireEnd}@
}
void @{\bf Promote}@(int j) {
    (taken, owner, owner_go) := LOCK_STATUS  @\label{rme:PromoteGoal1}@
    if taken = 0:
        active := ACTIVE
        if active $\neq$ 0:
            j := next(owner, active)
        if j $\neq \perp$:
            CAS(LOCK_STATUS, (0, owner, owner_go), 
                    (1, j, GO[j]))  @\label{rme:PromoteGoal1End}@
    (taken, owner, owner_go) := LOCK_STATUS @\label{rme:PromoteGoal2}@
    if taken = 1:
        *owner_go := TRUE  @\label{rme:PromoteGoal2End}@
}

\end{lstlisting}
\end{minipage}
\caption{$W$-port abortable RME algorithm}
\label{fig:w_port}
\end{figure}

Since we assume $W$-bit memory words, we are careful not to use unbounded, monotonically increasing counters,
which the Jayanti and Joshi lock does use.  Our algorithm's RMR bounds are in both the DSM and CC models,
whereas the Jayanti and Joshi lock has linear RMR complexity on the standard CC model.

\ifpaper
\subsection{Algorithm Walk-Through} \label{sec:w port walkthrough}
\else
\section{Algorithm Walk-Through} \label{sec:w port walkthrough}
\fi

\Cref{fig:w_port} presents the pseudo code of the algorithm.  We assume participating processes uses distinct ports
in the range $0,\dots,W-1$, so we refer to processes and ports interchangeably.  For simplicity, we present the
algorithm assuming dynamic memory allocation with safe reclamation~\cite{HP}.  In this environment, a process can
\emph{allocate} and \emph{retire} objects, and it is guaranteed that an allocation does not return a previously-retired
object if some process still has a reference to that object.  We show how to satisfy this assumption (with $O(D^2)$ static, pre-allocated
memory)
\ifCR
\else
in~\cref{remove-memory-restrict}.
\fi

Each process $p$ has a status word, $STATUS[p]$, and a pointer to a boolean spin variable, $GO[p]$. (In the DSM model,
a process allocates its spin variables from local memory, so that it can spin on them with $O(1)$ RMR cost.) The lock's state consists
of a $W$-bit word, $ACTIVE$, and a $\Theta(W)$-bit word, $LOCK\_STATUS$.  The $LOCK\_STATUS$ word holds a tuple
$(taken, owner, owner\_go)$, where $taken$ is a bit indicating if the lock is acquired by some process.  If $taken$ is set,
$owner$ is the id (port) of the lock's owner and $owner\_go$ points to the owner's spin variable.

The $STATUS$ word of each process $p$, initialized to $TRY$, indicates in which section the process is currently at.
This information is used by $Recover$ to steer $p$ to the right method when it arrives.  The $STATUS$ word changes
when completing $Try$ and entering the CS, when aborting during $Try$, when exiting the CS and executing $Exit$, and when $Exit$ completes.
Note that the \emph{Exit} method may be called as a subroutine during the \emph{Try} section's abort flow. In this case,
its operations are considered part of the \emph{Try} section (i.e., the subroutine call is to avoid putting a copy
of \emph{Exit}'s code in \emph{Try}). To distinguish these subroutine calls from when a process invokes \emph{Exit}
to exit the CS, we add an \emph{abort} argument to \emph{Exit}, which is \emph{FALSE} if and only if \emph{Exit} is
invoked to exit the CS (i.e., not as a subroutine).

In the normal (crash- and abort-free) flow, a passage of process $p$ proceeds as follows.  First, $p$ allocates its
spin variable, if it does not currently exist (lines~\ref{rme:checkSpin}--\ref{rme:checkSpinEnd}). Then $p$ flips its bit in the $ACTIVE$ word, but only if $p$'s bit
is not already set (lines~\ref{rme:Flip}--\ref{rme:FlipEnd}). This check avoids corrupting $ACTIVE$ when $p$ recovers from a crash.  Next, $p$ executes a $Promote$
procedure, which tries to pick some waiting process (possibly $p$) and make it the owner of the lock, if the lock is
currently unowned (line~\ref{rme:TryPromote}). Finally, $p$ begins spinning on its spin variable, waiting for an indication that it has become the lock owner (lines~\ref{rme:Spin}--\ref{rme:SpinEnd}).
Upon exiting the CS, $p$ clears its bit in $ACTIVE$ (again, only if the bit is currently set, to handle crash
recovery) (lines~\ref{rme:ExitClear}--\ref{rme:ExitClearEnd}). Then $p$ executes $Promote$ (line~\ref{rme:ExitPromote1}).  Performing this call will have no effect, since $p$ is still holding the lock,
which may appear strange, but is required in order to support the abort flow, as explained shortly.
Then, if $p$ is indeed the lock owner (another check useful only in the abort flow), it releases the lock by clearing the $taken$ bit in $LOCK\_STATUS$ (lines~\ref{rme:ExitCheckOwn}--\ref{rme:ExitCheckOwnEnd}).
Note that $p$ leaves the $owner$ and $owner\_go$ fields intact, for reasons described shortly.
Finally, $p$ executes $Promote$ again, to hand the lock off to some waiting process (line~\ref{rme:ExitPromote2}). It then retires its spin variable, clears its $GO$ pointer, and updates its $STATUS$ to $TRY$, thereby completing $Exit$ and thus its current passage and super-passage (line~\ref{rme:ExitRetire}--\ref{rme:ExitRetireEnd}).

If $p$ receives the abort signal while spinning in $Try$, it sets its $STATUS$ to $ABORT$, executes the $Exit$ method as
a subroutine, and returns
$FALSE$ (label~\ref{rme:TryAbort}--\ref{rme:TryAbortEnd}).  If $p$ crashes during the execution of $Exit$, $Recover$ will steer it to $Try$ once it recovers, at which
point it will again execute the $Exit$ method and return $FALSE$.  In the abort flow, the call to $Exit$ does
not modify $p$'s $STATUS$ (the $if$ is not taken, lines~\ref{rme:ExitStatus}--\ref{rme:ExitStatusEnd}).

The main goal of $Promote(j)$ is to promote some waiting process to be the lock owner, if the lock is currently unowned.
$Promote$ tries to promote one of the waiting processes (as specified by ACTIVE). If there is no such process, then $Promote$
tries to promote process $j$ if $j \neq \perp$, and does not promote any process otherwise (lines~\ref{rme:PromoteGoal1}--\ref{rme:PromoteGoal1End}).  A secondary goal of $Promote$ is that it signals the (current or newly promoted) owner by writing to its spin variable (lines~\ref{rme:PromoteGoal2}--\ref{rme:PromoteGoal2End}).
Picking a process to promote from among the waiting processes is done in a manner that guarantees starvation-freedom.
To this end, $Promote$ picks the next id whose bit is set in $ACTIVE$, when ids are scanned starting from the previous
owner's id (which, as described above, is written in $LOCK\_STATUS$) and moving up (modulo $W$).
(In the code, this is
specified as $next(owner,active)$.)  Having picked a process $q$ to promote, $Promote$ tries to update $LOCK\_STATUS$ to $(1,q,GO[q])$ using a single CAS. Finally, before completing, $Promote$ checks again if the lock is owned by some process $r$ (possibly $r \neq q$),
and if so, signals $r$ by writing $TRUE$ to $r$'s spin variable.

The reason for executing $Promote$ in $Exit$ before releasing the lock, and not only afterwards, is to handle a scenario
in which the lock owner $q$ has released the lock and $next(q, ACTIVE)=p$, so any process $r$ (possibly, but not necessarily, $q$)
executing $Promote$ tries to hand the lock to $p$. If now $p$ is signalled to abort, and did not also execute $Promote$
before departing, deadlock would occur.  By having $p$ call $Promote(p)$, we guarantee that either (1) some process
(possibly $p$) promotes $p$, so $p$'s $Exit$ call releases the lock before completing the abort; or (2) some process $r$ (possibly, but not necessarily $p$),
which does not observe $p$ in $ACTIVE$, updates $LOCK\_STATUS$ from $(0,q,G)$ to $(1,q',G')$. In the latter case, our
memory management assumption implies that $LOCK\_STATUS$ will not recycle to contain $(0,q,G)$ before every processes that
has read $(0,q,G)$ from $LOCK\_STATUS$ executes its CAS.  All such CASs, who are about to change $(0,q,G)$ to  $(1,p,GO[p])$ thus fail, so the lock does not get handed to $p$
and no deadlock occurs after it completes its abort.

\ifpaper
\subsection[Guaranteeing Starvation-Freedom] {Discussion: Guaranteeing Starvation-Freedom In the Presence of  Infinitely Many Aborts} \label{abortable not wait-free}
\else
\section[Guaranteeing Starvation-Freedom] {Discussion: Guaranteeing Starvation-Freedom In the Presence of  Infinitely Many Aborts} \label{abortable not wait-free}
\fi

As discussed in~\cref{sec:w port walkthrough}, a key idea in our algorithm is to invoke $Promote$ even before
releasing the lock, to handle the case in which the lock is about to be handed to an aborting process.
While simple, this is a subtle idea, because a different (more straightforward) approach to dealing with this issue
can lead to starvation.  We explain the issue by describing and analyzing a starvation problem in Jayanti and Joshi's
abortable RME algorithm~\cite{abortable-RME}.  The structure of our algorithm and of Jayanti and Joshi's algorithm is similar,
if one thinks of our $ACTIVE$ word and their min-array as abortable objects which (1) maintain the set of waiting processes
and (2) have some notion of the ``next in line'' waiting process, which becomes the lock owner.  (Jayanti and Joshi refer
to this object as a \emph{registry}.)  We describe the problem in the Jayanti and Joshi lock by contrasting
its behavior with our algorithm's.

Intuitively, starvation-freedom should follow from property (2) of the ``waiting room'' object, because every process
executing $Promote$ will eventually agree on the process $p$ to promote, which would then become the lock owner.
For this to be true, however, aborts need to be handled very carefully.
Phrased in our terminology, in the Jayanti and Joshi algorithm, a process $p$ that receives an abort signal starts
executing $Exit$, where it removes itself from the ``waiting room'' object.  Subsequently, if $LOCK\_STATUS=(0,o,os)$,
$p$ tries (using a single CAS) to update $LOCK\_STATUS$ from $(0,o,os)$ to $(0,p,GO[p])$.  In other words, $p$ tries
to make it look as if it had acquired the lock and immediately released it.  The motivation for this step is to fail
any $Promote$ that is about to make $p$ the lock owner, which if not handled, would result in deadlock.

This approach has the unfortunate side-effect of failing concurrent $Promote$s even if they are not about to
make $p$ the lock owner.  This can lead to an execution in which aborting processes prevent the lock from being
acquired, as described next.

Process $p_1$ arrives and enters the critical section.  Process $p_2,p_3,p_4$ arrive and enter the waiting room.
Now $p_1$ leaves the CS and executes $Exit$, which (in Jayanti and Joshi's algorithm) has a single $Promote$ call,
after releasing the lock.  Suppose the ``waiting room'' object indicates that $p_2$ should be the next lock owner.
Now, $p_1$ stops in its $Promote$ call, just before CASing $LOCK\_STATUS$ from $(0,p_1,*)$ to $(1,p_2,*)$.
Next, $p_3$ aborts, executes the $Exit$ code and successfully changes $LOCK\_STATUS$ to $(0,p_3,*)$.

As a result, $p_1$'s CAS in $Promote$ fails. $p_1$ completes its $Exit$ section and then returns to the Try section,
executes $Promote$, and stops just before CASing $LOCK\_STATUS$ from $(0,p_3,*)$ to $(1,p_2,*)$.
Now, $p_3$ proceeds to the $Promote$ call in $Exit$, stopping just before CASing $LOCK\_STATUS$ from $(0,p_3,*)$ to $(1,p_2,*)$.
We have reached a state in which $p_4$ is waiting, $LOCK\_STATUS$ is $(0,p_3,*)$, $p_1$ is in its Try $Promote$
and $p_3$ is in its $Exit$ promote, both about to CAS $LOCK\_STATUS$ from $(0,p_3,*)$ to $(1,p_2,*)$.

We continue as follows.  Now $p_4$ receives the abort signal, proceeds to execute $Exit$, and successfully changes
$LOCK\_STATUS$ from $(0,p_3,*)$ to $(0,p_4,*)$.  Consequently, the CAS of both $p_1$ and $p_3$ fails, so $p_1$
enters the waiting room, whereas $p_3$ departs the algorithm, returns, and stops in the Try $Promote$ before CASing
$LOCK\_STATUS$ from $(0,p_4,*)$ to $(1,p_2,*)$. As for $p_4$, it enters the Exit $Promote$ and stops before
CASing $LOCK\_STATUS$ from $(0,p_4,*)$ to $(1,p_2,*)$.  We have reached a similar situation as in the previous
paragraph, and can therefore keep repeating this scenario indefinitely. Throughout, $p_2$ keeps taking steps
in the waiting room, but will never enter the CS.
\ifpaper
\subsection{Proofs of RME Properties} \label{w-port proof}
\else
\section{Proofs of RME Properties} \label{w-port proof}
\fi

Whenever a process $p$ starts a super-passage in our algorithm, it allocates a fresh spin variable.
To avoid unbounded space consumption, the memory used for spin variables eventually has to be recycled,
i.e., an allocation by process $p$ can return a variable it previously used.
Our proofs assume that this recycling is done safely, namely, that an allocation of a new spin variable
does not return an object that is currently being referenced by some process.  (We show how to satisfy
this assumption using $O(D^2)$ static pre-allocated memory words
\ifCR
\else
in~\cref{remove-memory-restrict}.)
\fi

The above safe memory management assumption implies two properties that we use throughout the proofs.
First, that if a process $p$ is about to CAS $LOCK\_STATUS$ in $Promote$, and $LOCK\_STATUS$ has changed between
$p$ last reading it and executing the CAS, then the CAS will fail. This holds because $LOCK\_STATUS$ necessarily
contains a different $owner\_go$ value.  Second, that if $p$ sets the spin variable of $q$ to $TRUE$ and
$q$ has already started a new super-passage, then $q$ will never read that $TRUE$ value. This holds
because $q$ allocates a different spin variable for its new super-passage.

\ifCR

\else
Let $p$ be a process in super-passage $sp$ using port $k$.  For each passage in $sp$, we say $p$ is in its Try (respectively Critical or Exit) section from the time it executes the first memory operation of that section and until $p$ crashes or executes the last memory operation of that section.
Recall the subtle difference between $p$ being in its Exit section and $p$ executing the \emph{Exit} method. The latter can happen even if $p$ is not in its Exit section, due to receiving an abort signal in its Try section and calling \emph{Exit} as a subroutine. We say $p$ is in its Exit section if and only if $p$ invokes \emph{Exit} with $abort=FALSE$.

While $p$ is in its Try section, we define that $p$ is in the \emph{waiting room} from the time $t$ where $p$ starts the \emph{while} loop in the Try section and until time $t'$ where $p$ finishes its Try section or crashes.

If $p$ is in a super-passage $sp$ with port $k$, we define $spin_{p, k, sp}$ as the only spin variable $p$ is waiting on in its waiting room in this super-passage.  We use the notation $\&v$ for the memory address of $v$.

We refer to our $W$-port abortable RME algorithm as Algorithm $M$.  We proceed to prove that Algorithm $M$ satisfies the
desired RME properties. All claims and proofs assume that executions are well-formed.
Omitted claims and proofs appear in~\cref{appendix:w-port proof}.
\fi

\ifCR
\else
\begin{restatable}{claim}{statusNoSP}\label{status in no sp}
If at time $t$, no process is in a super-passage with port $k$, then $STATUS[k]=TRY$.
\end{restatable}

\begin{restatable}{claim}{statusTrySection}\label{status in try section}
If process $p$ is in super-passage $sp$ with port $k$ and $p$ is in the Try section, then $STATUS[k]=TRY$ or $STATUS[k]=ABORT$.
\end{restatable}

\begin{restatable}{claim}{statusInCS}\label{status in cs}
If process $p$ is in super-passage $sp$ with port $k$ and $p$ is in the CS, then $STATUS[k]=CS$.
\end{restatable}

\begin{restatable}{claim}{statusInExit}\label{status in exit section}
If process $p$ is in super-passage $sp$ with port $k$ and $p$ is in the Exit section, then $STATUS[k]=EXIT$.
\end{restatable}

\begin{restatable}{lemma}{recoverCorrect}\label{recover is correct}
Algorithm $M$ is well-behaved.
\end{restatable}

\begin{restatable}{claim}{StatusBeforePInCS}\label{CS_STATUS before p in CS}
If process $p$ is in super-passage $sp$ with port $k$ and is in the CS, then there was some time $t$ before $p$ first entered the CS in $sp$, at which $LOCK\_STATUS = (1, k, \&spin_{p, k, sp})$.
\end{restatable}

\begin{restatable}{claim}{StatusOnlyChangedByP}\label{CS_STATUS only changed by p}
If process $p$ is in super-passage $sp$ with port $k$, and at some time $t$ during $sp$, $LOCK\_STATUS = (1, k, \&spin_{p, k, sp})$, then $LOCK\_STATUS$ can only be changed by $p$ while executing the \emph{Exit} procedure.
\end{restatable}

\begin{restatable}{claim}{CSStatusInCS}\label{CS_STATUS in CS}
If process $p$ is in super-passage $sp$ with port $k$ and $p$ is in the CS, then $LOCK\_STATUS = (1, k, \&spin_{p, k, sp})$.
\end{restatable}

\begin{restatable}{lemma}{wPortME}\label{w-port ME}
Algorithm $M$ satisfies mutual exclusion.
\end{restatable}

\begin{restatable}{claim}{pointInPromote}\label{point in promote}
If process $p$ completes an execution of \emph{Promote} without crashing in the interval $[t, t']$ and $ACTIVE \ne 0$ throughout $[t, t']$, or if $p$ completes an execution of \emph{Promote(k)} without crashing in the interval $[t, t']$, then there is a time $t_0 \in [t, t']$ at which $LOCK\_STATUS=(1, *, *)$.
\end{restatable}

\begin{restatable}{claim}{CSStatusOneBeforeWaitingRoom}\label{CS_STATUS 1 before waiting room}
If process $p$ is in super-passage $sp$ with port k, $p$ sets the $k$-th bit in $ACTIVE$ to $1$ at time $t$, and subsequently reaches the waiting room at time $t'>t$, then there is a time $t_0 \in [t, t']$ at which $LOCK\_STATUS=(1, *, *)$.
\end{restatable}

\begin{restatable}{claim}{spinEventuallyTrue}\label{GO[k] eventually true}
If $LOCK\_STATUS$ changes to $(1, k, \&spin_{p, k, sp})$ then there is some process $p$ in its super-passage with port $k$ and either (1) $spin_{p, k, sp}$ eventually becomes $TRUE$, (2) $p$ aborts, or (3) there are infinitely many crashes in the execution.
\end{restatable}

\begin{restatable}{claim}{StatusChangesFromOneToZero}\label{CS_STATUS changes from 1 to 0}
If $LOCK\_STATUS$ changes to $(1, *, *)$ at time $t$, and if any process that enters the CS after $t$ eventually exits it, then there is a time $t'>t$ at which $LOCK\_STATUS$ changes to $(0, *, *)$ or there are infinitely many crashes in the execution.
\end{restatable}

\begin{restatable}{claim}{StatusChangesFromZeroToOne}\label{CS_STATUS changes from 0 to 1}
If at some point of the execution $LOCK\_STATUS=(0, *, *)$, and $ACTIVE \ne 0$, then $LOCK\_STATUS$ will change to $(1, *, *)$ or there are infinitely many crashes in the execution.
\end{restatable}

\begin{restatable}{claim}{atMostWTimes}\label{at most w times}
If a process $p$ is in its super-passage $sp$ with port $k$ and is in its waiting room, and $p$ does not abort, then $LOCK\_STATUS$ can change to $(1, *, *)$ at most $W$ times before changing to $(1, k, \&spin_{p, k, sp})$.
\end{restatable}
\begin{proof}
We define the \emph{distance of the $j$-th bit from the $i$-th bit}, denoted $d(i, j)$, as follows.
If $j>i$ then $d(i,j) = j-i$.  If $i>j$ then $d(i,j) = W-i+j$.  Informally, $d(i,j)$ is the number of times we count up
from $i$, modulo $W$, to obtain $j$.  We number the bits in $ACTIVE$ from $1$ to $W$ such that the least significant bit is numbered $1$ and the most significant bit numbered $W$. Assume that $p$ is now in its waiting room in super-passage $sp$ with port $k$, and that $LOCK\_STATUS$ is  $(1, o, os)$. Assume now that $LOCK\_STATUS$ changes again to $(1, o', os')$. The process that successfully CASed $LOCK\_STATUS$ previously
reads $(1, o, os)$ from it, and then reads $a=ACTIVE$, implying that it sees the $k$-th bit set in $a$. Let $o'=next(o, a)$ be the next owner of the lock. It follows from the definition of $d(o,k)$ that either $o'=k$ or that $d(o', k) < d(o, k)$. Since the distance of any bit $i$ from any bit $j$ is at most $W$ and at least 0, the aforementioned change can happen at most $W$ times before $next()$ returns $k$.
\end{proof}

\begin{lemma}\label{w-port SF}
Algorithm $M$ satisfies starvation-freedom.
\end{lemma}
\begin{proof}
Assume there are finitely many crash steps in the execution.  Assume process $p$ is in a super-passage $sp$ with port $k$, in which $p$ does not abort.
Since there are finitely many crashes, then $p$ eventually reaches the waiting room in $sp$.
From~\cref{CS_STATUS 1 before waiting room}, there is a time $t$ before $p$ reaches the waiting room at which $LOCK\_STATUS=(1, *, *)$.
Since there are finitely many crashes and every process does not stop taking steps while in a super-passage,
every process that enters the CS after $t$ also reaches its Exit section.  Thus, from~\cref{CS_STATUS changes from 1 to 0}, there is a point $t' > t$ at which $LOCK\_STATUS$ changes to $(0, *, *)$.  Because bit $k$ in $ACTIVE$ is set, \cref{CS_STATUS changes from 0 to 1} implies that as long as $p$ is in its waiting room, at some point $t'' > t'$, $LOCK\_STATUS$ changes again to $(1, *, *)$.  We can repeat this argument as long as $p$ is in its waiting room and does not abort. However, from~\cref{at most w times}, after at most $W$ times, $LOCK\_STATUS$ becomes $(1, k, \&spin_{p, k, sp})$. At this point, by~\cref{GO[k] eventually true}, $spin_{p, k, sp}$ eventually turns $TRUE$. Subsequently, $p$ enters the $CS$ in a finite number of its own steps.
\end{proof}

\begin{restatable}{theorem}{wPortMain}\label{w port proprties}
If every execution of Algorithm $M$ is well-formed, then Algorithm $M$ satisfies mutual exclusion, bounded abort, starvation-freedom, CS re-entry, wait-free CS re-entry, wait-free exit, and super-passage wait-free exit. The passage complexity of Algorithm $M$ in both the $CC$ and $DSM$ models is $O(1)$ and the super-passage complexity is $O(1+F)$. (Assuming, for the DSM model, that process memory allocations return local memory.) The space complexity of the algorithm is $O(D^2)$.
\end{restatable}
\fi
\ifpaper
\else
\section{Memory Management}\label{remove-memory-restrict}

Here we show how to bound our algorithm's space consumption. 
The basic idea is that when a process $p$ starts
a super-passage and ``allocates'' a fresh spin variable, the allocation will be satisfied from a pre-allocated static
array of $2D+1$ spin variables maintained for $k$, the port used by $p$ in this super-passage.  (Thus, the overall memory
consumption is $O(D^2)$.)  The challenges are
(1) how to identify an entry in port $k$'s array that is \emph{safe} and not being referenced by any other process (these entries
may have previously been used by other processes accessing port $k$, and so there may be active processes that hold references to them);
and (2) how to find the entry with $O(1)$ RMRs.

To identify a safe entry, we use standard ideas from safe memory reclamation~\cite{HP}.  The high-level idea is that
when process $p$ executing $Promote$ reads $(*,k,v)$ from $LOCK\_STATUS$, it announces that it now has a reference to $v$ by writing a pointer to $v$ in a new $REFERENCED$ array.  This announcement enables any process using port $k$ to avoid reusing this spin variable $v$, as long as $p$ might still be referencing it.
It could, however, be the case that $p$ reads $(b,k,v)$ and delays before writing $v$ to $REFERENCED$, causing some process $q$ using port $k$ to decide to reuse $v$. To
handle this problem, $p$ must \emph{validate} that $LOCK\_STATUS$ still contains $v$ after writing the announcement.
If the validation succeeds, $p$ is considered to hold a reference to $v$. In such a case, $p$ knows that both $REFERENCED$ contains $v$ and $LOCK\_STATUS=(b,k,v)$ held at the same point in time, so an allocation by some process $q$ with port $k$ that starts after $p$ has a reference to $v$ will not reuse $v$.

What if $p$'s validation fails?  We observe that in such a case, $p$ can simply abandon its $Promote$ operation.
If $p$ observes $b=1$, then by~\cref{CS_STATUS only changed by p}, the validation failing implies that $q$ has exited the
lock, so there is no need to signal $q$.  Similarly, if $p$ observes $b=0$, the validation failing implies that $p$'s
later CAS is meant to fail (\cref{sec:w port walkthrough}).

\Cref{fig:w_port_with_memory} shows the full algorithm: the code from~\Cref{fig:w_port} extended with safe memory reclamation,
as described above, as well as the procedures used to ``allocate'' a spin variable.  In the rest of the section, we describe how this ``allocation'' is implemented in $O(1)$ RMRs.  We essentially use the lock-free allocation scheme of Aghazadeh et al.~\cite[Section 2.3]{memory}, except that we make it wait-free using the aforementioned observation that a failed validation can be abandoned instead of retried.

Each port $k$ is associated with a \emph{FREE} queue that holds $O(D)$ free variables that may be ``allocated'' safely; i.e., it is
guaranteed that no other process holds a reference to these variables.  Process $p$ ``allocates'' a spin variable by calling the \emph{GetIndex} procedure (line~\ref{memory:call_getindex}) which pops one spin variable from the \emph{FREE} queue of its port (lines~\ref{memory:getindex_start}--\ref{memory:getindex_end}).  When $p$ finishes the super-passage (due to aborting or
exiting the CS), it retires its current spin variable $v$ (line~\ref{memory:call_retire}).  However, $v$ can be pushed back into the \emph{FREE} queue
only after it is verified that no other process holds a reference to it.  The key idea is to perform this verification lazily, by checking one
entry of the $REFERENCED$ array in each super-passage that uses port $k$ (as part of the \emph{Retire} call).  To facilitate this, we maintain a reference counter in $v$.  The reference counter of $v$ is initialized to 1 when $v$ is retired (line~\ref{memory:refcount_one}),
and is incremented each time a reference to $v$ is observed while scanning $REFERENCED$ (lines~\ref{memory:refcount_inc_start}--\ref{memory:refcount_inc_end}).  In addition, each such
reference counter update is undone (by decrementing the counter) once $D$ entries of $REFERENCED$ have been scanned
since the update was made (lines~\ref{memory:refcount_dec_start}--\ref{memory:refcount_dec_end}).
All together, this means that if $v$'s reference counter ever becomes 0, then no reference to it was observed in $REFERENCED$
and $v$ is put back into the port's \emph{FREE} queue (lines~\ref{memory:refcount_zero_start}--\ref{memory:refcount_zero_end}).  The \emph{Retire} procedure implements this algorithm by tracking all
pending counter updates in two (per-port) queues of size $D$---a \emph{RETIRED} queue for initial counter updates and an
\emph{OBSERVED} queue for counter increments---and pushing and popping one item from each queue in each invocation of
\emph{Retire} (lines~\ref{memory:track_spin_start}--\ref{memory:track_spin_end}). Since at any time, these queues can contain at most $2D$ spin variables out of the $2D+1$ spin variables associated
with the port, we are guaranteed that a port's \emph{FREE}
queue always contains at least one spin variable when a super-passage starts.

\begin{figure}[!bt]
\hspace{-10ex}
\noindent\begin{minipage}{.60\textwidth}
\begin{lstlisting}[firstnumber=1]
 @{\bf ACTIVE}@: int // initially 0
 @{\bf STATUS}@: array of $W$ status words // initially all TRY
 @{\bf GO}@: array of $W$ pointers to Spin variables
 // initially all $\perp$
 @{\bf LOCK\_STATUS}@: struct {bool, port_id, pointer to Spin}
 // initially (0, 0, $\perp$)

struct Spin {
    int refcount
    bool val
    int port
}

  @{\bf FREE}@: array of $W$ queues of Spin variables
 // queue in index $k$ is initialized with $2W+1$ Spin
 variables, each with refcount=0, val=FALSE, and port=$k$
 @{\bf RETIRED}@: array of $W$ queues of pointers to Spin vars
 // initially each queue contains $W$ $\perp$ values
  @{\bf OBSERVED}@: array of $W$ queues of pointers to Spin vars
 // initially each queue contains $W$ $\perp$ values
 @{\bf COUNTER}@: array of $W$ integers
 // initially all 0
 @{\bf REFERENCED}@ array of $W$ pointers to Spin variables
 // initially all $\perp$

void @{\bf Retire}@(int k, Spin spin) {
    spin.refcount := 1 @\label{memory:refcount_one}@
    RETIRED[k].push(spin)

    ann := REFERENCED[COUNTER[k]]
    OBSERVED[k].push(ann)
    if ann $\neq$ $\perp$ and ann.port = k: @\label{memory:refcount_inc_start}@
        old_refcount := ann.refcount
        ann.refcount := old_refcount + 1 @\label{memory:refcount_inc_end}@

    for l in (RETIRED[k], OBSERVED[k]): @\label{memory:track_spin_start}@
        x := l.pop()
        if x $\neq$ $\perp$: @\label{memory:refcount_dec_start}@
            old_refcount := x.refcount
            x.refcount := old_refcount - 1 @\label{memory:refcount_dec_end}@
            if x.refcount = 0: @\label{memory:refcount_zero_start}@
                x.val := FALSE
                FREE[k].push(x) @\label{memory:refcount_zero_end}@
    @\label{memory:track_spin_end}@
    old_counter_p := COUNTER[k]
    COUNTER[k] := old_counter_p + 1 mod W
}

Spin @{\bf GetIndex}@(k) {@\label{memory:getindex_start}@
    spin := FREE[k].pop()
    return spin
} @\label{memory:getindex_end}@

status @{\bf Recover}@(int k) {
    if STATUS[k] = EXIT:
        return EXIT
    if STATUS[k] = CS:
        return CS
    return TRY
} 
\end{lstlisting}
\end{minipage}\hfill
\begin{minipage}{.60\textwidth}
\begin{lstlisting}[firstnumber=last]
void @{\bf Try}@(int k) {
    if STATUS[k] = ABORT:
        Exit(k, TRUE)
        return FALSE

    if GO[k] = $\perp$:
        if got abort signal:
            STATUS[k] := ABORT
            Exit(k, TRUE)
            return FALSE
        GO[k] := GetIndex(k) @\label{memory:call_getindex}@

    if k-th bit in ACTIVE is 0:
        FAA(ACTIVE, $2^k$)

    Promote(k, $\perp$)

    while GO[k].value = FALSE:
        if got abort signal:
            STATUS[k] := ABORT
            Exit(k, TRUE)
            return FALSE

    STATUS[k] := CS
    return TRUE
}

void @{\bf Exit}@(int k, bool aborting) {
    if aborting = FALSE:
        STATUS[k] = EXIT

    if k-th bit in ACTIVE is 1:
        FAA(ACTIVE, $-2^k$)

    Promote(k, k)

    (taken, owner, owner_go) := LOCK_STATUS
    if taken = 1 and owner = k:
        CAS(LOCK_STATUS, (1, owner, owner_go),
            (0, owner, owner_go))

    Promote(k, $\perp$)

    if (GO[k] $\neq \perp$ ):
        Retire(k, GO[k]) @\label{memory:call_retire}@
        GO[k] := $\perp$

    STATUS[k] := TRY
}

void @{\bf Promote}@(int k, int j) {
    (taken, owner, owner_go) := LOCK_STATUS
    REFERENCED[k] := owner_go

    if (taken, owner, owner_go) = LOCK_STATUS:
        if taken := 0:
            active := ACTIVE
            if active $\neq$ 0:
                j := next(owner, active)
            if j $\neq \perp$:
                CAS(LOCK_STATUS, (0, owner, owner_go),
                    (1, j, GO[j]))
    REFERENCED[k]:= $\perp$

    (taken, owner, owner_go) := LOCK_STATUS
    REFERENCED[k] := owner_go
    if (taken, owner, owner_go) := LOCK_STATUS:
        if taken := 1:
            owner_go.value := TRUE
    REFERENCED[k] := $\perp$
}



\end{lstlisting}
\end{minipage}
\vspace{-3ex}
\caption{$W$-port abortable RME algorithm with memory reclamation}
\label{fig:w_port_with_memory}
\end{figure}

It is straightforward to make the above scheme recoverable: all operations during allocation are done to private variables,
so we can simply hold these variables, as well as the program counter of the \emph{GetIndex}/\emph{Retire} procedures,
in the NVRAM, so that after a crash the process can resume its execution from the point it crashed. It is then safe to re-execute
the operation (read/write) that was about to be performed, since the operation accesses a process-private variable, and all writes in those procedures are idempotent. The queues being used can also easily be made recoverable by constructing them in such a way that all writes are idempotent and again holding all variables in NVRAM and using the program counter to return to the point of the crash. For brevity, \Cref{fig:w_port_with_memory} shows  \emph{Retire} and \emph{GetIndex} without these details.

\fi

\ifpaper
\section{Tournament Tree} \label{chap:tree}
\else
\fi

A \emph{tournament tree} lock, referred to as the \emph{main} lock, is constructed by statically arranging
multiple $D$-port RME algorithms, referred to as \emph{node} locks, in a $D$-ary tree with $N$ leaves
(we assume $D \leq W$).
Each leaf is uniquely associated with a process.
To acquire the main lock, a process competes to acquire each lock on the path from its leaf to the root, until it wins
at the root and enters the main lock's CS.
To release the main lock, the process descends from the root to its leaf, releasing each node lock on the path.
In this section, we present our tournament tree algorithm.

Our algorithm has two distinguishing features:
(1) that its super-passage RMR complexity is additive in $F$, the number of crashes, and not multiplicative;
and (2) that it satisfies \emph{super-passage wait-free exit} (SP-WF-Exit), i.e., a process releasing the main
lock is guaranteed to complete some execution of $Exit$ after a finite number of its own steps (including crashes).

Our algorithm's super-passage RMR complexity is $O(FR + B \log_D N)$, where $R$ and $B$
are the recovery cost and passage complexity of the node lock, respectively. In comparison, prior trees have
super-passage complexity of $O(F(R + B\log_D N))$.  Obtaining our bound is simple:
a process just needs to write its location in the tree to NVRAM, so that upon crash recovery, it can resume
from there instead of starting to walk up or down the tree from scratch.  We suspect that this simple optimization
was not performed in prior tournament trees because their node lock has $R = \log_D N = O(\frac{\log N}{\log \log N})$
and $B=O(1)$, so directly returning to the node at which the crash occurred does not asymptotically improve complexity.
With our $W$-port RME algorithm, however, $R=B=O(1)$, so being
additive in $F$ is asymptotically better, and would not be obtained using prior tournament trees.

\ifpaper
The problem of obtaining SP-WF-Exit highlights the difficulty of composing
recoverable locks.  The issue is that a process in the main lock is composing critical sections of the node
locks, which creates the problem of how recovery of the main and node locks interact.  In the model of
prior work~\cite{journal,sublog-first}, a process crashing in the main lock's exit section attempts to re-acquire
the main lock upon recovering.  As a result, the process might now block in some node lock's entry section,
which violates SF-WF-Exit for the main lock.
We address this problem by carefully modeling RME algorithms in a way that facilitates composition (\cref{chap:Model}).
Instead of assuming how a process participates in the algorithm (i.e., cycling through entry, CS, exit), we
model the RME algorithm as an object whose $Recover$ procedure informs the process where it
crashed in the super-passage.  This approach allows client algorithms, composing the lock, to decide how to
proceed.  Our model allows a process returning to lock $x$ after crashing in the main lock to realize that
it had completed an $x$-super-passage \emph{and not start a new one}.  Consequently, our tournament tree avoids
the problems described above and satisfies SP-WF-Exit.
\else
The problem of obtaining SP-WF-Exit highlights the difficulty of composing
recoverable locks.  The issue is that a process in the main lock is composing critical sections of the node
locks, which creates the problem of how recovery of the main and node locks interact.  In the model of
Golab and Ramaraju~\cite{journal}, a process crashing in the main lock's exit section attempts to re-acquire
the main lock upon recovering.  As a result, the process might now block in some node lock's entry section,
which violates SP-WF-Exit for the main lock.  In Golab and Hendler's tournament tree~\cite{sublog-first},
only a high-level description of the algorithm is provided, stating that ``all recovery actions are performed
in the recovery protocol of the $k$-port algorithm, and for that reason the arbitration tree does not require a
separate recovery protocol.''  It is thus not clear what happens if a process $p$, which is exiting the main
lock, crashes immediately after exiting some node lock $x$ but before managing to record this fact in memory.
We can only assume that on recovery, $p$ eventually returns to $x$.  However, from the perspective of the
lock $x$, $p$ has completed an $x$-super-passage before crashing.  When $p$ returns to $x$ it therefore
enters $x$'s entry section, where it competes with other processes and may block---which violates SP-WF-Exit
for the main lock.

We address the above problem by carefully modeling RME algorithms in a way that facilitates composition (\cref{chap:Model}).
Instead of assuming how a process participates in the algorithm (i.e., cycling through entry, CS, exit), we
model the RME algorithm as an object whose $Recover$ procedure informs the process where it
crashed in the super-passage.  This approach allows client algorithms, composing the lock, to decide how to
proceed.  Our model allows a process returning to lock $x$ after crashing in the main lock to realize that
it had completed an $x$-super-passage \emph{and not start a new one}.  Consequently, our tournament tree avoids
the problems described above and satisfies SP-WF-Exit.
\fi

We present detailed pseudo code and prove all of the algorithm's properties.
\ifCR
\else
Omitted proofs appear in~\cref{appendix:tree proof}.
\fi

\begin{figure}[t]
\noindent\begin{minipage}{.57\textwidth}
\begin{lstlisting}[firstnumber=1]
 @{\bf STATUS}@: array of $N$ status words // initially all TRY.
 @{\bf CURR\_NODE}@: array of $N$ nodes
 // initially $CURR\_NODE[i]$ is the $i$-th leaf.

void @{\bf Try}@(int pid) {
    if STATUS[pid] = ABORT:
        Exit(k, TRUE)
        return FALSE
    node = CURR_NODE[pid]
    while STATUS[pid] $\neq$ CS or node $\neq$ root: @\label{tree:start_ascend}@
        if node is the j-th child of node.parent, 
        then set k to j
        if node.Recover(k) = TRY: @\label{tree:start_ascend_recover}@
            node.Try(k)@\label{tree:end_ascend_recover}@
        if received abort signal: @\label{tree:start_abort}@
            STATUS[pid] := ABORT
            Exit(pid, TRUE)
            return FALSE  @\label{tree:end_abort}@
        if node = root:
            break
        node := node.parent
        CURR_NODE[pid] := node @\label{tree:end_ascend}@
    STATUS[pid] := CS
    return TRUE
}

\end{lstlisting}
\end{minipage}\hfill
\begin{minipage}{.45\textwidth}
\begin{lstlisting}[firstnumber=last]
void @{\bf Exit}@(int pid, bool aborting) {
    if aborting = FALSE:
        STATUS[pid] = EXIT
    node := CURR_NODE[pid]
    while TRUE: @\label{tree:start_descend}@
        if node is the j-th child of node, 
        then set k to be j
        if node.Recover(k) $\neq$ TRY:  @\label{tree:start_descend_recover}@
            node.Exit(k, FALSE) @\label{tree:end_descend_recover}@
        if node = LEAF:
            break
        node := node.child(k)
        CURR_NODE[pid] := node @\label{tree:end_descend}@
    STATUS[pid] := TRY
}
status @{\bf Recover}@(int pid) {
    if STATUS[pid] = EXIT:
        return EXIT
    if STATUS[pid] = CS:
        return CS
    return TRY
} 
\end{lstlisting}
\end{minipage}
\vspace{-3ex}
\caption{The Tournament Tree}
\label{fig:tournament_tree}
\vspace{-3ex}
\end{figure}

\ifpaper
\subsection{Algorithm Walk-Through}
\else
\section{Algorithm Walk-Through}
\fi

\Cref{fig:tournament_tree} shows the pseudo code of the algorithm.  Each node has immutable $parent$ and $child$ pointers
(as mentioned before, the tree structure is static). The $parent$ of root is $\perp$, as are all $child$ pointers of a leaf node.
Each process is statically assigned to a leaf based on its id ($pid$).  Each node contains a $D$-port abortable RME lock.

Similarly to our $W$-port algorithm, each process $p$ has a status word, $STATUS[p]$, which is used by the main lock's
$Recover$ procedure.  Each process has a $current\_node$ pointer.

In $Try$, a process walks the path from its leaf to the root, acquiring each node lock along the way (lines~\ref{tree:start_ascend}--\ref{tree:end_ascend}).
In each such lock, it uses a statically assigned port, corresponding to the number of the child from which it climbed
into the node.  After successfully acquiring the lock at node $x$, process $p$ writes $x$ to $current\_node[p]$ (line~\ref{tree:end_ascend}).
This allows $p$ to return to $x$ if it crashes, instead of having to start from scratch and climb the entire path
again.  The $Exit$ flow is symmetric, with $p$ releasing each lock along the path back to the leaf, and updating
$current\_node[p]$ after each lock release (lines~\ref{tree:start_descend}--\ref{tree:end_descend}).  In both entry and exit flows, $p$ always execute node lock's $Recover$
procedure before entering that lock's Try or Exit section. This allows $p$ to behave correctly after crash
recovery: on its way up (respectively, down) it will not execute $Enter$ (respectively, $Exit$) on the same node lock twice (lines~\ref{tree:start_ascend_recover}--\ref{tree:end_ascend_recover}, respectively lines~\ref{tree:start_descend_recover}--\ref{tree:end_descend_recover}).

To support aborts, process $p$ checks the abort signal after acquiring each node lock (lines~\ref{tree:start_abort}-\ref{tree:end_abort}).  If an abort was signalled,
$p$ starts executing the main lock's exit code to descend from the current node back to its leaf, releasing the
node locks it holds along the way.  (Similarly to the $W$-port algorithm, an aborting process execute $Exit$ as
a subroutine; it does not formally enter the main lock's exit section).  The algorithm correctly supports aborts
because if an abort is signalled while $p$ is in some node lock's $Try$ execution, it is guaranteed to complete in
a finite number of its own steps. Subsequently, it will execute the main lock's abort handling code in a constant
number of its own steps.

\ifpaper
\else
\ifpaper
\else
\section {Proofs of RME Properties} \label{tree proof}
\fi

We refer to our main abortable RME algorithm as Algorithm $Tree(M)$, where $M$ is some $D$-port RME algorithm that satisfies mutual exclusion, deadlock-freedom, bounded abort, starvation-freedom, CS re-entry, wait-free CS re-entry and super-passage wait-free exit. We define a process $p$ to be in its Try, Critical or Exit section of $Tree(M)$ similarly to the way defined in\cref{w-port proof}. We proceed to prove that Algorithm $Tree(M)$ satisfies the
desired RME properties.  All claims and proofs assume that executions are well-formed.

The main idea is to prove that if an execution of the $Tree(M)$ algorithm is well-formed, then the execution of every node lock $L$ in the tree is also well-formed, and so every node lock maintains all of its properties.  We then show that the properties of $Tree(M)$ follow from that.

\begin{lemma}\label{tree recover is correct}
Algorithm $Tree(M)$ is well-behaved.
\end{lemma}
\begin{proof}
Analogous to the proof of~\cref{recover is correct}, since both algorithms update the $STATUS[p]$ in the same way, and the \emph{Recover} procedures of both algorithms are identical.
\end{proof}

\begin{restatable}{claim}{treeCorrectUse}\label{tree correct use}
Each execution of a $D$-port node lock $M$ in the $Tree(M)$ algorithm is well-formed.
\end{restatable}
\begin{proof}
We need to show that for every node lock $L$ in the tree, the projection of $Tree(M)$'s execution on $L$ is well-formed, i.e., that it satisfies the six properties of~\cref{defn:well-formed}.  We show this using induction over the height of the tree.

The base case considers the leaves.
Properties~\ref{prop:port} (constant port usage) and~\ref{prop:sp} (no concurrent super-passages) hold trivially since the ports are statically assigned, and no two processes are assigned the same leaf with the same port.  Let $L$ be some leaf.  We proceed by induction over $E_L$, the subsequence of the execution occurring at $L$.  Assume that $E_L$ is well-formed until time $t$.  Then:

\begin{itemize}
\item Property~\ref{prop:recover} (Recover invocation): From the code, $p$ always completes an execution of $L.Recover$ before either invoking $L.Try$ or $L.Exit$, or ascending/descending through $L$'s node.

\item Property~\ref{prop:try} (Try invocation):
Let $t'$ be the first time after $t$ at which some process $p$ invokes $L.Try$.  
From the code of $Tree(M)$, $p$ invokes $L.Try$ only after completing \emph{L.Recover}.  If at $t'$, $p$
either invokes $L.Try$ for the first time in $E_L$, or for the first time after completing $L.Recover$ since $p$'s last $L$-super-passage,
then we are done.  (This is the beginning of an $L$-super-passage by $p$.)  Otherwise, $p$ has previously invoked $L.Try$ in its
current $L$-super-passage, which starts before $t'$.  Since the execution is well-formed up to $t'$ and $L$ is well-behaved,
it follows that $p$'s prior crash step was in $L.Try$.

\item Property~\ref{prop:cs} (CS invocation):
Let $t'$ be the first time after $t$ at which some process $p$ enters $L$'s CS.  
From the code of $Tree(M)$, $p$ enters $L$'s CS in one of two cases.  First, if at $t'$, $p$ completes $L.Try$ which returns $TRUE$.
(If $L.Try$ returns $FALSE$, since the execution is well-formed up to $t'$ and $L$ is well-behaved, it follows that an abort is
signalled, and thus $p$ does not advance to $L$'s CS.)  In this case, we are done.  Otherwise, $p$ completes $L.Recover$ which
returns $CS$.  Since the execution is well-formed up to $t'$ and $L$ is well-behaved, it follows that $p$'s prior crash step was
in $L$'s CS.

\item Property~\ref{prop:exit} (Exit invocation):
Let $t'$ be the first time after $t$ at which some process $p$ invokes $L.Exit$.  
From the code of $Tree(M)$, $p$ invokes $L.Exit$ only after completing \emph{L.Recover} that returns $r \neq TRY$.
First, observe that it cannot be the case the $p$ invokes $L.Exit$ after a crash-free execution from an invocation of $L.Try$
that returns $FALSE$ (i.e., after aborting while trying to acquire $L$).  Since the execution is well-formed up to $t'$ and $L$ is well-behaved, $L.Exit$ returns $TRY$ in such a case, and so $p$ does not invoke $L.Exit$.  If $p$ is in $L$'s CS, we are done.
Otherwise, since the execution is well-formed up to $t'$ and $L$ is well-behaved,
it follows that $p$'s prior crash step was in $L$'s Exit section.
\end{itemize}

Next, let $L$ be a non-leaf lock in the tree.  Assume the claim is correct for all locks in the sub-tree of $L$.

\begin{itemize}
\item Property~\ref{prop:port} (constant port usage): Immediate, since processes use statically assigned ports to access each lock in the tree.
\item Property~\ref{prop:sp} (no concurrent super-passages): From the code, process $p$ accesses $L$ with port $k$ only if $p$ is in the CS of the $k$-th child of $L$.  Since the execution of $L$'s $k$-th child is well-formed, it satisfies mutual exclusion. Thus, there are no two concurrent $L$-super-passages $sp_i$ and $sp_j$ of processes $p_i \ne p_j$ with the same port.
\item Properties~\ref{prop:recover}--\ref{prop:exit}: Follow from similar reasoning as the leaf case.
\end{itemize}
\end{proof}

\begin{lemma}\label{tree me}
Algorithm $Tree(M)$ satisfies mutual exclusion.
\end{lemma}
\begin{proof}
If some process $p$ is in the CS of $Tree(M)$, then it must be in the CS of the root lock. \cref{tree correct use} implies that the root lock satisfies mutual exclusion, so the claim follows.
\end{proof}

\begin{lemma} \label{tree sf}
Algorithm $Tree(M)$ satisfies starvation-freedom.
\end{lemma}
\begin{proof}
Assume there are finitely many crash steps in the execution.
Suppose, towards a contradiction, that process $p$ executes infinitely many steps without entering the $CS$ of $Tree(M$)
and without receiving an abort signal.
Since there are finitely many crashes, it follows that $p$ executes infinitely many steps in $L.Try$ for some node lock $L$,
without receiving an abort signal. This contradicts the fact that $L$ is starvation-free.
\end{proof}

\begin{lemma} \label{tree bounded abort}
Algorithm $Tree(M)$ satisfies bounded abort.
\end{lemma}
\begin{proof}
Immediate from the fact that each node lock $L$ satisfies bounded abort and wait-free exit.
\end{proof}

\begin{restatable}{theorem}{treeMain}\label{tree port proprties}
Let $M$ be some $D$-port RME algorithm that satisfies mutual exclusion, starvation-freedom, bounded abort, CS re-entry, wait-free CS re-entry, and super-passage wait-free exit.
If every execution of Algorithm $Tree(M)$ is well-formed, then Algorithm $Tree(M)$ satisfies mutual exclusion, starvation-freedom, bounded abort, CS re-entry, wait-free CS re-entry, wait-free exit, and super-passage wait-free exit.
The passage complexity of Algorithm $Tree(M)$ in both the CC and DSM models is $O(B \log_D N)$, where $B$ is the node lock passage
complexity.  The super-passage complexity of Algorithm $Tree(M)$ is $O(FR + B \log_D N)$, where $R$ is the recovery complexity of the node lock.
The space complexity of Algorithm $Tree(M)$ is $O(S \frac{N}{D} \log_D N)$, where $S$ is the space complexity of the node lock.
\end{restatable}
\begin{proof}
Mutual exclusion follows from~\cref{tree me}.  Starvation-freedom follows from~\cref{tree sf}. Bounded abort follows from~\cref{tree bounded abort}.
Since $Tree(M)$ is well-behaved and its \emph{Recover} is wait-free, then CS re-entry and wait-free CS re-entry are satisfied.
Because each node lock satisfies wait-free exit, both wait-free exit and super-passage wait-free exit are satisfied.

The height of the tree is $O(log_w N)$ and a process ascends (respectively, descends) the height of the tree to complete \emph{Try} (respectively, \emph{Exit}).  Since each node lock $L$ has passage complexity $B$, the passage complexity of $Tree(M)$ is $O(B \log_D N)$.

In the worst case, following a crash, a recovering process returns to a node lock $L$ where it just completed a passage and
executes the recovery code passage, so each crash incurs $O(1+R)$ RMRs.  Thus, the algorithm's super-passage complexity is $O(FR + B \log_D N)$.

Since there are $O(\frac{N}{D})$ leaves in the tree, and the height of the tree is $O(\log_D N)$, there are $O(\frac{N}{D} \log_D N)$ nodes in the tree.  The overall space complexity of the algorithm is therefore $O(S \frac{N}{D} \log_D N)$.
\end{proof}

\fi
\ifpaper
\section{Adaptive Transformation} \label{chap:transform}
\else
\fi

\begin{figure}[!b]
\noindent\begin{minipage}{.47\textwidth}
\input{./code_adaptive_transformation_a}
\end{minipage}\hfill
\begin{minipage}{.58\textwidth}\input{./code_adaptive_transformation_b}

\end{minipage}\hfill
\vspace{-3ex}
\caption{Adaptive Transformation}
\label{fig:adaptive_transformation}
\end{figure}

We now present our generic adaptivity transformation, which transforms any abortable RME algorithm $L$ whose
RMR complexity depends only on $N$ into an abortable RME algorithm whose RMR complexity also depends on
the \emph{point contention}~\cite{contention1, contention2}, $K$, which is the number of processes executing the algorithm concurrently with the process going through the super-passage.
We show how to transform an abortable RME algorithm with passage complexity $B < W$, super-passage complexity $B^*$, and space complexity
$S$, into an abortable RME algorithm with passage complexity $O(\min(K,B))$, super-passage complexity
$O(K +F)$ if $K < B$ or $O(B^* + F)$ otherwise, and space complexity $O(S + N + B^2)$.

The transformation is essentially a fast-path/slow-path construction, where the fast path is our $W$-port
abortable RME algorithm and the slow path is the original lock $L$.  A process $p$ attempts to capture
port $k=0,\dots,W-1$ so it can use it in the fast path lock. Each such capture attempt is performed with
CAS, and hence incurs an RMR. The idea is that if $p$ fails to capture a port, then another process $q$
succeeds. Therefore, if $p$ fails to capture any port, the point contention is $> W$. In this case, $p$
gives up and enters the slow path. The fast path and slow paths are synchronized with a 2-port abortable
RME lock, again implemented with our lock (\cref{chap:w_port}).

We present detailed pseudo code and prove all of the algorithm's properties.
\ifCR
\else
Omitted proofs appear in~\cref{appendix:transform proof}.
\fi 
\ifpaper
\subsection{Algorithm Walk-Through}
\else
\section{Algorithm Walk-Through}
\fi

\Cref{fig:adaptive_transformation} presents the transformed algorithm's pseudo code.  The transformed
algorithm uses three auxiliary abortable RME locks: a $slow\_path$ lock, which is an $N$-process base lock
being transformed into an adaptive lock, and $fast\_path$ as well as $2\_rme$ locks, both of which are
instances of our $D$-port abortable RME (\cref{chap:w_port}).  The $fast\_path$ instance uses
$D=B$ and the $2\_rme$ instance uses $D=2$.

The algorithm maintains $K\_OWNERS$, an array of $B$ words (initially all $\perp$) through which processes
in the entry section try to capture ports to use in the fast-path lock  (lines~\ref{transform:get_k_start}--\ref{transform:get_k_end}).  Each process maintains a $CURR\_K$ variable to
store the next port the process attempts to capture, or its captured port (once it captures one).
To capture a port, process $p$ scans $K\_OWNERS$, using CAS at each slot $k$ in an attempt to capture
port $k$. If $p$ captures port $k$, it enters the fast-path lock using that port.  Overall, if $p$ reaches
slot $k$ in $K\_OWNERS$, then $k$ other processes have captured ports $0,...,k-1$.
If $p$ reaches the end of $K\_OWNERS$ and fails to capture a port, it enters the slow-path lock (lines~\ref{transform:slow_start}--\ref{transform:slow_end}). Regardless
of which lock $p$ ultimately enters, it invokes that lock's $Recover$ method first, to correctly handle the case
in which $p$ is recovering from a crash.

We use the 2-RME lock to ensure mutual exclusion between the owners of the fast-path and the slow-path.
Once $p$ acquires its lock, it enters the 2-RME lock from the right (respectively, left) if it is on the
fast-path (respectively, slow-path). In the 2-RME lock, $p$ takes on a unique right/left id, corresponding to
its direction of entry. Once $p$ acquires the 2-RME lock, it enters the CS (lines~\ref{transform:2_rme_start}--\ref{transform:2_rme_end}).

In the exit section, $p$ releases the 2-RME lock (lines~\ref{transform:2_rme_release_start}--\ref{transform:2_rme_release_end}) and then the fast-path or slow-path lock, as appropriate (lines~\ref{transform:release_start}--\ref{transform:release_end}).
After releasing the fast-path lock, $p$ releases its port (line~\ref{transform:release_k})    . These steps are done carefully to avoid having
$p$ return to the fast-path lock after crashing with the same port that is now being used by another process.

To handle aborts, if $p$ receives a FALSE return value from some $Enter$ execution, it executes the transformed
lock's exit code (which, as a byproduct, releases $p$'s port if it has one). Subsequently, $p$ completes the abort.

\ifpaper
\else
\ifpaper
\else
\section {Proofs of RME Properties} \label{transform proof}
\fi

Let $M$ be an $N$-port abortable RME algorithm that satisfies mutual exclusion, deadlock-freedom, bounded abort, starvation-freedom, CS re-entry, wait-free CS re-entry, and super-passage wait-free exit.  We refer to our transformation applied to $M$ as Algorithm $Adapt(M)$.
We define a process $p$ to be in the Try, Critical or Exit section of $Adapt(M)$ analogously to definition in~\cref{w-port proof}. We proceed to prove that Algorithm $Adapt(M)$ satisfies the desired RME properties.  All claims and proofs assume that executions are well-formed.

The main idea is to prove that if an execution of the $Adapt(M)$ is well-formed then the execution of the three auxiliary locks used in the algorithm are also well-formed, and so they maintain their RME properties. The correctness of $Adapt(M)$ follows from that.

\begin{lemma}\label{transform recover is correct}
Algorithm $Adapt(M)$ is well-behaved.
\end{lemma}
\begin{proof}
Analogous to the proof of~\cref{recover is correct}, since both algorithms update the $STATUS[p]$ using the same flow, and the $Recover$ procedures of both algorithms are identical.
\end{proof}

\begin{restatable}{claim}{slowPathCorrect}\label{slow path correct}
Any execution of the $slow\_path$ lock in $Adapt(M)$ is well-formed.
\end{restatable}
\begin{proof}
We need to show that the projection of $Adapt(M)$'s execution over the $slow\_path$ lock is well-formed, i.e., that it satisfies the five properties of~\cref{defn:well-formed}.  Properties~\ref{prop:port} constant port usage) and~\ref{prop:sp} (no concurrent super-passages) hold trivially, since
$slow\_path$ is an $N$-port, and process ids are used as port numbers.  We prove Properties 1--3 by induction over $E_s$, the subsequence of the execution occurring at the $slow\_path$ lock.  Assume that $E_s$ is well-formed until time $t$.  Then:

\begin{itemize}
\item Property~\ref{prop:recover} (Recover invocation): From the code, $p$ always completes an execution of $slow\_path.Recover$ before either invoking $slow\_path.Try$ or $slow\_path.Exit$. Moreover, if $p$ is in a $slow\_path$-super-passage, then $PATH[p]=SLOW$, and so $p$ always invokes $slow\_path.Recover$ after recovering from a crash.

\item Property~\ref{prop:try} (Try invocation):
Let $t'$ be the first time after $t$ at which some process $p$ invokes $slow\_path.Try$.  
From the code of $Adapt(M)$, $p$ invokes $slow\_path.Try$ only after completing \emph{slow\_path.Recover}.  If at $t'$, $p$
either invokes $slow\_path.Try$ for the first time in $E_s$, or for the first time after completing $slow\_path.Recover$ since $p$'s last $slow\_path$-super-passage,
then we are done.  (This is the beginning of a $slow\_path$-super-passage by $p$.)  Otherwise, $p$ has previously invoked $slow\_path.Try$ in its current $slow\_path$-super-passage, which starts before $t'$.  Since the execution is well-formed up to $t'$ and $slow\_path$ is well-behaved, it follows that $p$'s prior crash step was in $slow\_path.Try$.

\item Property~\ref{prop:cs} (CS Invocation):
Let $t'$ be the first time after $t$ at which some process $p$ enters $slow\_path$'s CS.  
From the code of $Adapt(M)$, $p$ enters $slow\_path$'s CS in one of two cases.  First, if at $t'$, $p$ completes $slow\_path.Try$ which returns $TRUE$.
In this case, we are done.  Otherwise, $p$ completes $slow\_path.Recover$ which
returns $CS$.  Since the execution is well-formed up to $t'$ and $slow\_path$ is well-behaved, it follows that $p$'s prior crash step was
in $L$'s CS.

\item Property~\ref{prop:exit} (Exit invocation):
Let $t'$ be the first time after $t$ at which some process $p$ invokes $slow\_path.Exit(FALSE)$.  
From the code of $Adapt(M)$, $p$ invokes $slow\_path.Exit$ only if $PATH[p]=SLOW$ and after completing \emph{slow\_path.Recover} that returns $r \neq TRY$.
First, observe that it cannot be the case the $p$ invokes $slow\_path.Exit$ after a crash-free execution from an invocation of $slow\_path.Try$
that returns $FALSE$ (i.e., after aborting while trying to acquire $slow\_path$), because then it invokes $Adapt(M)$'s \emph{Exit} procedure with an \emph{aborting} argument of $TRUE$.  If $p$ is in $slow\_path$'s CS, we are done.
Otherwise, since the execution is well-formed up to $t'$ and $slow\_path$ is well-behaved,
it follows that $p$'s prior crash step was in $slow\_path$'s Exit section.
\end{itemize}
\end{proof}

\begin{restatable}{claim}{fastPathCorrect}\label{fast path correct}
Any execution of the $fast\_path$ lock in $Adapt(M)$ is well-behaved.
\end{restatable}
\begin{proof}
The proof of Properties~\ref{prop:recover}--\ref{prop:exit} is analogous to the proof of the $slow\_path$ lock.
Constant port usage and no concurrent super-passages are satisfied since, from the code, if $p$ invokes
$fast\_path.Recover(k)$, $fast\_path.Try(k)$ or $fast\_path.Exit(k)$, it must be the case that $K\_OWNERS[k]=p$. Further, if $K\_OWNERS[k]=p$, it can only be changed to $\perp$, and that can only happen by $p$ in \emph{Exit}, which can only happen after $p$ finished its $fast\_path.Exit(k)$, implying $p$ finished its $fast\_path$-super-passage with port $k$ and will no longer access it.
\end{proof}

\begin{restatable}{claim}{TwoRmeCorrect}\label{2 rme correct}
Any execution of the $2\_rme$ lock in $Adapt(M)$ is well-behaved.
\end{restatable}
\begin{proof}
The proof of Properties~\ref{prop:recover}--~\ref{prop:exit} is analogous to the proof of the $slow\_path$ lock.

 Property~\ref{prop:port} (Constant port usage): Assume a process $p$ sets $SIDE[p]=RIGHT$, this means $PATH[pid]=FAST$, which means there is some $k$ such that $K\_OWNERS[k]=p$ and $CURR\_K[p]=k$. These can only change in $Adapt(M)$'s Exit section, after $p$ has finished the $2\_rme$ Exit section. So if $p$ crashes while accessing $2\_rme$ then it either executes the $Adapt(M)$ $Entry$ section again, where it will again set $SIDE[p]=RIGHT$, or $Adapt(M)$'s CS or Exit section, where it will not change $SIDE[p]$. Similarly, if $p$ sets $SIDE[p]=LEFT$ then there is no $k$ such that $K\_OWNERS[k]=p$ and $CURR\_K[p]=B$, causing $p$ to skip the while loop in $Adapt(M)$'s Try section if it executes it again in this super-passage.  Thus $p$ will not change $PATH[pid]$, and consequently will also set $SIDE[p]$ to $LEFT$ again, and so constant port usage is satisfied.

 Property~\ref{prop:sp} (No concurrent super-passages): In order for some process $p$ to set $PATH[pid]=RIGHT$ (respectively, $LEFT$), it must be in the CS of the $fast\_path$ (respectively, $slow\_path$) lock, and since $p$ sets it back to $\perp$ before releasing the $fast\_path$ (respectively, $slow\_path$, and both of them satisfy mutual exclusion, then no concurrent super-passages is satisfied.
\end{proof}

\begin{restatable}{claim}{TransformationAllRmeCorrect}\label{transformation well behaved}
Any execution of each of the auxiliary locks in $Adapt(M)$ is well-behaved.
\end{restatable}
\begin{proof}
Follows from~\cref{slow path correct,fast path correct,2 rme correct}.
\end{proof}

\begin{lemma}\label{transform me}
Algorithm $Adapt(M)$ satisfies mutual exclusion.
\end{lemma}
\begin{proof}
If some process $p$ is in the CS of $Adapt(M)$, then it must be in the CS of the $2\_rme$ lock.~\cref{transformation well behaved} implies that the $2\_rme$ lock satisfies mutual exclusion, so the claim follows.
\end{proof}

\begin{lemma}\label{transform sf}
Algorithm $Adapt(M)$ satisfies starvation-freedom.
\end{lemma}
\begin{proof}
Assume there are finitely many crash steps in the execution.
Suppose, towards a contradiction, that process $p$ executes infinitely many steps without entering the $CS$ of $Adapt(M$)
and without receiving an abort signal.
Since there are finitely many crashes, it follows that $p$ executes infinitely many steps in the $Try$ of some auxiliary lock $M$, which contradicts the fact that these locks are starvation-free.
\end{proof}

\begin{lemma} \label{transform bounded abort}
Algorithm $Adapt(M)$ satisfies bounded abort.
\end{lemma}
\begin{proof}
Assume there are finitely many crash steps in the execution.
Assume $p$ is in its $Try$ section and receives an abort signal.  After a finite number of its own steps, $p$ either returns or reaches $M.Try()$ for some auxiliary lock $M$. Since the execution of $M$ is well-formed, and $p$ has abort signalled, after a finite number of $p$'s steps, $M.Try()$ returns, and $p$ executes the Exit code of the $Adapt(M)$, which is wait-free.
\end{proof}

\begin{claim} \label{point contention acquire port}
Denote by $start(p)$ the time at which $p$ starts its latest super-passage. If at time $t$, process $p$ fails to acquire port $i$, then there exist a time $t_i^p \in (start(p), t]$ such that at $t_i^p$ there is a set $X_i^p$ of $i$ processes, such that for each process $q \in X_i^p$ the following holds: $t_i^p > start(q)$ and $q$ successfully acquired a port number $\leq i$ by time $t$.
\end{claim}
\begin{proof}
By induction. For $i=0$, pick $t_0^p=t$ and $X_0^p=\{q\}$, where $q$ is the process that causes $p$'s CAS on port 0 to fail.
For the inductive step, assume the claim is true up to some $i$. Consider the process $q$ that causes $p$'s CAS on port $i+1$ to fail. There are two possible cases:

\begin{enumerate}
\item $start(q) < t_i^p$. If $q\in X_i^p$ then $q$ has previously acquired a port $\leq i$. However, this contradicts the fact that $q$ has now acquired port $i+1$ in the same super-passage. Thus, $q \notin X_i^p$. Then we define $t_{i+1}^p = t_i^p$ and $ X_{i+1}^p =  X_i^p \cup \{ q\}$.

\item $start(q) \geq t_i^p$. The fact that $q$ acquired port $i+1$ means that $q$ failed to acquire port $i$ before time $t$. Consider the time $t_i^q > start(q)$ and set $X_i^q$ from the induction hypothesis applied to $q$. We have $t_i^p \leq start(q) < t_i^q  \leq t$. Also, $p \notin X_i^q$, since $p$ does not acquire any port by time $t$. Thus, we can define $t_{i+1}^p = t_i^q$ and $X_{i+1}^p = X_i^q \cup \{q\}$.
\end{enumerate}
\end{proof}

\begin{claim} \label{point contention complexity}
If some process $p$ fails to acquire some port $i$, then the point contention this process experiences is at least $i$.
\end{claim}
\begin{proof}
Immediate from~\cref{point contention acquire port}
\end{proof}

\begin{restatable}{theorem}{transformMain}\label{transformation proprties}
Let $M$ be some $N$-port RME algorithm that satisfies mutual exclusion, starvation-freedom, bounded abort, CS re-entry, wait-free CS re-entry, and super-passage wait-free exit, with passage complexity $B < W$, super-passage complexity $B^*$, and space complexity $S$.
If every execution of Algorithm $Adapt(M)$ is well-formed, then Algorithm $Adapt(M)$ satisfies mutual exclusion, starvation-freedom, bounded abort, CS re-entry, wait-free CS re-entry, wait-free exit, and super-passage wait-free exit. The passage complexity of Algorithm $Adapt(M)$ in both the CC and DSM models is $O(\min(K, B))$, the super-passage complexity is $O(K + F)$ if $K<B$ and $O(B^* + F)$ otherwise. The space complexity of Algorithm $Adapt(M)$ is $O(S + N + B^2)$.
\end{restatable}
\begin{proof}
Mutual exclusion follows from~\cref{transform me}.  Starvation-freedom follows from~\cref{transform sf}.
Bounded abort follows from~\cref{transform bounded abort}.
Since $Adapt(M)$ is well-behaved and its \emph{Recover} is wait-free, then CS re-entry and wait-free CS re-entry are satisfied.
Because each auxiliary lock satisfies wait-free exit, both wait-free exit and super-passage wait-free exit are satisfied.

Every time $p$ starts a new super-passage, $CURR\_K[p]=0$, since it is reset in the previous super-passage's Exit code (or is in the initial state).  Thus, $p$ tries to acquire a port number in the range of $0$ and $B$, and from~\cref{point contention complexity} it can fail at most $K$ times. If $p$ succeeds, then the RMR cost of acquiring and releasing the $fast\_path$ lock and $2\_rme$ lock is $O(1)$.  If $p$ fails, it  performs additional $O(B)$ RMRs to acquire and release the $slow\_path$ and $2\_rme$ locks.  Overall, the total RMR passage complexity is $O(min(K, B))$.

If $p$ crashes while trying to acquire the $fast\_path$ lock, then when recovering, it does not need to try to acquire a lock starting from $0$ again, but it starts from the last port it tried to acquire before crashing.  So if $K<B$, $p$ eventually succeeds in acquiring a port in $O(k + F)$ RMRs. If $K>B$, then it already performed $O(B + F)$ RMRs, and will need to perform additional $O(B^*)$ RMRs to acquire and release the $slow\_path$ and $2\_rme$ locks.

The space complexity of the $Adapt(M)$ is $O(S + N + B^2)$ since the space complexity of the $slow\_path$ is $O(S)$, the space complexity of the $fast\_path$ is $O(B^2)$, the space complexity of the $2\_rme$ lock is $O(1)$, and the transformation itself uses additional $O(N)$ space.
\end{proof}

\fi

\section{Putting It All Together \& Conclusion}
Let $T$ be the RME algorithm obtained by instantiating our tournament tree (\cref{chap:tree})
with our $W$-port abortable RME algorithm (\cref{chap:w_port}).  Then $T$'s RMR passage
complexity is $O(\log_W N) < W$, super-passage complexity is $O(\log_W N + F)$ and space complexity is $O(NW\log_W N)$. We can therefore apply the transformation
of~\cref{chap:transform} to $T$, obtaining our main result:

\begin{theorem}
There exists an abortable RME with $O(\min(K,\log_W N))$ RMR passage complexity, $O(F + \min(K,\log_W N))$
RMR super-passage complexity, and $O(NW\log_W N)$ space complexity where $K$ is the point contention, $W$ is the memory word size, $N$ is the
number of processes, and $F$ is the number of crashes in a super-passage.
\end{theorem}

Many questions about ME properties in the context of RME remain open, and we are far from understanding
how the demand for recoverability affects the possibility of obtaining other desirable properties and
their cost.  Can the sublogarithmic RMR bounds be improved using only primitives supported in hardware,
such as FAS and FAA?
It is known that a weaker crash model facilitate better bounds~\cite{system-crash},
but is relaxing the crash model necessary?
What, if any, is the connection between RME and abortable mutual exclusion?  Both problems involve
a similar concept, of a process ``disappearing'' from the algorithm, and for both problems, the best
known RMR bounds (assuming standard primitives) are $O(\frac{\log N} {\log \log N})$.  Can a formal
connection between these problems be established?

\bibliography{paper}

\ifCR
\else
\appendix

\section {Proofs Omitted From Section \ref{w-port proof}} \label{appendix:w-port proof}

\begin{restatable}{claim}{completeExit}\label{complete exit}
If process $p$ completes a super-passage $sp$ with port $k$, then $p$ completes an execution of the \emph{Exit} procedure and does not subsequently write to any memory location before $sp$ ends.
\end{restatable}
\begin{proof}
By definition, $p$ finishes $sp$ either by completing an \emph{Exit} section, which means completing an execution of \emph{Exit}, or if $p$'s \emph{Try} section returns $FALSE$. From the code, a \emph{Try} section only returns $FALSE$ immediately after completing \emph{Exit}.
\end{proof}

\begin{restatable}{claim}{variablesInit}\label{variables init}
When a process $p$ starts a super-passage $sp$ with port $k$, then $GO[k]=\perp$, the $k$-th bit of $ACTIVE$ is $0$, and $STATUS[k]=TRY$.
\end{restatable}
\begin{proof}
By induction on the number of $M$-super-passages using port $k$.  The base case is the first $M$-super-passage that uses port $k$. In this case, $GO[k]$, the $k-th$ bit of $ACTIVE$ and $STATUS[k]$ are initialized to their default values which are $\perp$, $0$, and $TRY$, respectively.
For the inductive case, consider the last $M$-super-passage by some process $p$ that used port $k$. From~\cref{complete exit}, when $p$ completes that super-passage, it completes an execution of the \emph{Exit} procedure, which resets $GO[k]$, the $k-th$ bit of $ACTIVE$, and $STATUS[k]$ to their initial values.
\end{proof}

\statusNoSP*
\begin{proof}
Immediate from~\cref{variables init}.
\end{proof}

\statusTrySection*
\begin{proof}
Recall that $p$ is in the Try section from the first operation of $Try(k)$ and until its last operation. If this is the first call to $M.try(k)$ in $sp$, then from~\cref{variables init}, $STATUS[k]=TRY$. From code inspection, $STATUS[k]$ can only change to $ABORT$ or to $CS$ in the Try section. If $STATUS[k]$ changes to $CS$, then by definition, $p$ is no longer in the Try section. If $STATUS[k]$ changes to $ABORT$, it can only later be changed by the \emph{Exit} procedure. Since $STATUS[k]=ABORT$, the first condition in \emph{Exit} does not hold, and so $STATUS[k]$ can only be changed back to $TRY$ when the \emph{Exit} procedure completes.  Subsequently, $p$ either completes the Try section by returning $FALSE$ or crashes.  If $p$ crashes after setting $STATUS[k]$ to $TRY$, then upon recovery, $p$ will change $STATUS[k]$ back to $ABORT$ and execute the \emph{Exit} procedure again, since the abort signal remains signalled.  This can repeatedly happen until eventually $p$ does not crash and completes the Try section.
Thus, $STATUS[k]$ is always $TRY$ or $ABORT$ during this time period.
\end{proof}

\statusInCS*
\begin{proof}
By definition, $p$ is in the CS from the time that $STATUS[k]=CS$, or if $p$ re-enters the CS after crashing inside it.  Therefore,
once $STATUS[k]$ becomes $CS$, $p$ does not invoke $M.Try(k)$ in $sp$ again, and so $STATUS[k]$ can only be changed when $p$ exits the CS and invokes $M.Exit(k)$. The first operation of \emph{Exit} changes $STATUS[k]$ to $EXIT$, which by definition is when $p$ exits the CS.
\end{proof}

\statusInExit*
\begin{proof}
By definition, $p$ is in the Exit section from when it changes $STATUS[k]$ to $EXIT$ in \emph{Exit}.  This change only happens if $p$ executes \emph{Exit} when exiting the CS (i.e., not in the abort flow), since the execution is well-formed.  Subsequently, $STATUS[k]$ can only be changed again in $sp$ by $p$ to $TRY$ at the end of \emph{Exit}.  By definition, this is when $p$ completes the Exit section and $sp$ ends.
\end{proof}

\recoverCorrect*
\begin{proof}
Immediate from \cref{status in no sp,status in try section,status in cs,status in exit section} and from the simple flow of the $Recover$ procedure.
\end{proof}

\begin{restatable}{claim}{newSpinVariable}\label{new spin variable}
If process $p$ is in super-passage $sp$ with port $k$,
and $p$ reaches the waiting room, then $p$ initializes $GO[k]$ to point to a new spin variable, before the $k$-th bit of $ACTIVE$ is set.
\end{restatable}
\begin{proof}
From~\cref{variables init}, when $p$ starts the super-passage, $GO[k]$ is $\perp$, and the $k$-th bit of $ACTIVE$ is $0$. Thus, before $p$ set the $k$-th bit to $1$ using FAA, there is an execution of \emph{Try} in which $p$ observes $GO[k]=\perp$ and so initializes $GO[k]$.
\end{proof}

\begin{restatable}{claim}{spinDoesntChange}\label{spin variable does not change}
If process $p$ is super-passage $sp$ with port $k$ sets $GO[k]$ to some value other than $\perp$ at time $t$, then $p$ does not change $GO[k]$ to any value except $\perp$ in $sp$.
\end{restatable}
\begin{proof}
From the code, once $p$ sets $GO[k]$ to some value other than $\perp$ in \emph{Try}, $GO[k]$ can only change when $p$ executes \emph{Exit}.
$p$ invokes \emph{Exit} either due exiting the CS or due to aborting.  If $p$ exits the CS, then since the execution is well-formed, $p$ does not execute \emph{Try} again afterwards in $sp$.  If $p$ aborts, then since the abort signal remains signalled in $sp$, $p$ never executes the step that initializes $GO[k]$ in \emph{Try} again.
\end{proof}

\begin{restatable}{claim}{spinWellDefined}\label{spin well defined}
If process $p$ reaches the waiting room in super-passage $sp$ with port $k$, then $spin_{p, k, sp}$ exists, and if $GO[k] \ne \perp$, then $GO[k]=\&spin_{p, k, sp}$.
\end{restatable}
\begin{proof}
From~\cref{new spin variable}, $p$ initializes $GO[k]$ to $spin_{p, k, sp}$ before reaching the waiting room in $sp$. From~\cref{spin variable does not change}, $GO[k]$ cannot change to any value except $\perp$ in $sp$.
\end{proof}

\begin{restatable}{claim}{spinIsTrue}\label{spin is true}
If process $p$ is in super-passage $sp$ with port $k$ and $p$ is in the $CS$, then $spin_{p, k, sp}=TRUE$.
\end{restatable}
\begin{proof}
Since the execution is well-formed, the first time in $sp$ at which $p$ reaches the CS is when it sets $STATUS[k]=CS$ in the Try section.
 This write can only happen if $p$ finished the while loop in the waiting room, which occurs when $*GO[k]=TRUE$. By~\cref{spin well defined}, this means that $spin_{p, k, sp}=TRUE$. This is the only write to $spin_{p, k, sp}$ in $sp$, so this variable stays $TRUE$ for the entire CS of $p$ in $sp$.
\end{proof}

\StatusBeforePInCS*
\begin{proof}
If $p$ is in its CS, then $spin_{p, k, sp}=TRUE$ by~\cref{spin is true}. From the code, a spin variable can only be changed to $TRUE$ in \emph{Promote} by some process $q$ (possibly $q=p$) that reads $LOCK\_STATUS=(1, *, \&spin_{p, k, sp})$. Similarly from the code, a process that writes
$LOCK\_STATUS=(1, *, \&spin_{p, k, sp})$ must actually write $LOCK\_STATUS=(1, k, \&spin_{p, k, sp})$.
\end{proof}

\StatusOnlyChangedByP*
\begin{proof}
From the code, $LOCK\_STATUS$ can only be changed from $(1, k, \&spin_{p, k, sp})$ to some other value in $sp$ by the \emph{Exit} procedure invoked with a port argument of $k$.  Since the execution is well-formed, such an invocation can only be performed by $p$.
\end{proof}

\CSStatusInCS*
\begin{proof}
From~\cref{CS_STATUS before p in CS}, there was a time $t$ before $p$ entered the CS at which $LOCK\_STATUS=(1, k, \&spin_{p, k, sp})$. By~\cref{CS_STATUS only changed by p}, this value can only be changed by $p$ executing the \emph{Exit} procedure.  Since the execution is well-formed, $p$ can only change $LOCK\_STATUS$ in its Exit section.  Thus, if $p$ is in the CS, $LOCK\_STATUS=(1, k, \&spin_{p, k, sp})$.
\end{proof}

\wPortME*
\begin{proof}
Assume that two processes, $p_i \ne p_j$, are in super-passages, $sp_i$ and $sp_j$, with ports $k_i \ne k_j$.
Assume towards a contradiction that both $p_i$ and $p_j$ are in the CS at the same time $t$. From~\cref{CS_STATUS in CS}, at time $t$,
$LOCK\_STATUS=(1, k_i, \&spin_{p_i, k_i, sp_i})=(1, k_j, \&spin_{p_j, k_j, sp_j})$,
implying $k_i=k_j$, which is a contradiction.
\end{proof}

\pointInPromote*
\begin{proof}
There are two possible scenarios for $p$'s execution of \emph{Promote}: When $p$ first reads $LOCK\_STATUS$, it either reads $LOCK\_STATUS=(1, *, *)$ or $LOCK\_STATUS=(0, *, *)$. In the first case, we are done.  Otherwise, if $p$ reads $LOCK\_STATUS=(0, *, *)$, it tries to execute a CAS to change $LOCK\_STATUS$ from $(0, *, *)$ to $(1, *, *)$, either because $ACTIVE \ne 0$ or because $j \ne \perp$. $p$'s CAS either succeeds or fails due to another process changing $LOCK\_STATUS$ to $(1, *, *)$. In any case, there is a time $t_0 \in [t, t']$ at which $LOCK\_STATUS=(1, *, *)$.
\end{proof}

\CSStatusOneBeforeWaitingRoom*
\begin{proof}
Since $p$ reaches the waiting room, it completes an execution of $Promote$ in some interval $T \subseteq [t, t']$ without crashing. This execution occurs after the $k$-th bit of $ACTIVE$ is set. From the code, only $p$ can clear the $k$-th bit of $ACTIVE$, and it can do so
only after reaching the waiting room.  Thus, throughout the interval $T$,  $ACTIVE \ne 0$. From~\cref{point in promote},
$LOCK\_STATUS=(1, *, *)$ at some time $t_0 \in T$.
\end{proof}

\spinEventuallyTrue*
\begin{proof}
Assume there are finitely many crash steps in the execution.  Assume $LOCK\_STATUS$ is changed to $(1, k, \&spin_{p, k, sp})$ by process $q$ (possibly $q=p$). From the code, this change occurs when $q$ executes the \emph{Promote} procedure.  Because there are finitely many crashes, $q$ eventually completes a \emph{Promote} call. (The reason is that since the execution is well-formed, if $q$ crashes inside \emph{Promote}, it returns to the same section of the lock and eventually executes \emph{Promote} again.) When $q$ completes \emph{Promote}, it either observes $LOCK\_STATUS=(1, k, \&spin_{p, k, sp})$ and sets $spin_{p, k, sp}=TRUE$, or it reads some other value from $LOCK\_STATUS$. In the latter case, by~\cref{CS_STATUS only changed by p}, $LOCK\_STATUS$ was changed by $p$ while executing the \emph{Exit} procedure.  The invocation of \emph{Exit} can happen if $p$ either aborts or exits the CS. If $p$ exits the CS, \cref{spin is true} implies that $spin_{p, k, sp}$ was set to $TRUE$.
\end{proof}

\begin{restatable}{claim}{changedToOnePNotFinished}\label{CS_STATUS 1 p yet to finish exit}
Consider process $p$ in super-passage $sp$ with port $k$.  If $LOCK\_STATUS$ is changed to $(1, k, \&spin_{p, k, sp})$ at time $t$, then at time $t$, $p$ has executed the FAA in \emph{Try} and has not completed the first \emph{Promote} in \emph{Exit} in $sp$.
\end{restatable}
\begin{proof}
Let $q$ be the process that changes $LOCK\_STATUS$ to $(1, k, spin_{p, k, sp})$.  There are two cases:
If $q=p$: From the code, $p$ changes $LOCK\_STATUS$ in \emph{Promote}.  Thus, from the code, time $t$ is after $p$'s FAA in \emph{Try}.
Moreover, time $t$ cannot occur in $p$'s invocation of \emph{Promote}($\perp$) in the \emph{Exit} procedure. In this invocation,
$p$ can update $LOCK\_STATUS$ to $(1,j,*)$ only if the $j$-th bit of   $ACTIVE$ is set. From the code, however, the $k$-th bit of ACTIVE is clear at this point.
If $q \ne p$: Let $t_1$ be the time when $p$ executes the FAA in \emph{Try}, $t_2$ be the time when $q$ reads $ACTIVE$ for the last time before changing $LOCK\_STATUS$, $t_3$ be the time when $q$ changes $LOCK\_STATUS$.  We first show $t_1 < t_3$.  Since $q$
changes $LOCK\_STATUS$ to $(1,k,spin_{p, k, sp})$ then $q$ observes the $k$-th bit of $ACTIVE$ set. Thus, $t_1 < t_2 < t_3$.  Next,
if $p$ never completes the first \emph{Promote} in \emph{Exit} in $sp$, we are done.  Otherwise, let $t_4$ be the time when $p$ executes the FAA in \emph{Exit} and $t_5$ be the time when $p$ finishes executing \emph{Promote}(k) in \emph{Exit}.  We show $t_3 < t_5$ by contradiction.
We have $t_2 < t_4$, since $q$ observes the $k$-th bit of $ACTIVE$ set.  Thus, $t_1 < t_2 < t_4 < t_5 < t_3$. Since $q$'s successful CAS of $LOCK\_STATUS$ succeeds, and $q$ reads $LOCK\_STATUS$ before reading $ACTIVE$, we have that $LOCK\_STATUS=(0, *, *)$ throughout $[t_2, t_3]$. But $p$ executed \emph{Promote}($k$) in $[t_4, t_5] \subset [t_2, t_3]$, and so $p$'s \emph{Promote}($k$) observes $LOCK\_STATUS=(0, *, *)$ and so $p$ should successfully CAS $LOCK\_STATUS$ to $(1, k, \&spin_{p, k, sp})$ by $t_5$, which is a contradiction.
\end{proof}

\StatusChangesFromOneToZero*
\begin{proof}
Assume there are finitely many crash steps in the execution. If $LOCK\_STATUS$ changes to $(1, k, \&spin_{p, k, sp})$ at time $t$, then by~\cref{CS_STATUS 1 p yet to finish exit}, process $p$ is in super-passage $sp$ with port $k$ and has not completed the first \emph{Promote} in \emph{Exit} in $sp$. By~\cref{GO[k] eventually true}, after $t$ either $spin_{p, k, sp}$ eventually becomes $TRUE$ or $p$ aborts.
In either case, $p$ eventually completes \emph{Exit}. As mentioned, $p$ does not complete \emph{Promote}($k$) in \emph{Exit}
before time $t$. Thus, $p$ changes $LOCK\_STATUS$ to $(0, k, \&spin_{p, k, sp})$ after time $t$, since that event is executed after
\emph{Promote}($k$) returns.
\end{proof}

\begin{restatable}{claim}{StatusChangedToZeroChangedToOne}\label{CS_STATUS changed to 0 then changes back}
If at time $t$, $ACTIVE \ne 0$ and $LOCK\_STATUS$ changes to $(0, *, *)$, then $LOCK\_STATUS$ eventually changes to $(1, *, *)$, or there are infinitely many crashes in the execution.
\end{restatable}
\begin{proof}
Assume there are finitely many crash steps in the execution. Let $p$ be the process that changes $LOCK\_STATUS$ to $(0, *, *)$ at time $t$.
Then $p$ is executing the \emph{Exit} procedure and still needs to execute \emph{Promote}($\perp$). Since the execution is well-formed, $p$ eventually completes an execution of \emph{Promote}($\perp$).  Let $t'$ be the time where $p$ completes its execution of \emph{Promote}($\perp$). If $ACTIVE \ne 0$ throughout $[t, t']$, then by~\cref{point in promote}, $LOCK\_STATUS$ changes to $(1, *, *)$ in $[t, t']$. Otherwise,
if $ACTIVE$ changes to $0$ during $[t, t']$, then there is some time $t_0 \in [t, t']$ at which some process $q$ executes $FAA$ in \emph{Exit}.
Since there are finitely many crashes, $q$ eventually completes an execution of \emph{Promote}($k$) at some time $t_0' > t_0 > t$. Thus, by~\cref{point in promote}, $LOCK\_STATUS$ eventually changes to $(1, *, *)$ after $t$.
\end{proof}

\begin{restatable}{claim}{LockChangesStatusChangesToOne}\label{LOCK changes then CS_STATUS changes to 1}
If at time $t$, $LOCK\_STATUS=(0, *, *)$ and $ACTIVE$ changes from 0 to a non-zero value, then $LOCK\_STATUS$ eventually changes to $(1, *, *)$, or there are infinitely many crashes in the execution.
\end{restatable}
\begin{proof}   
Assume there are finitely many crash steps in the execution.  Since at time $t$, $ACTIVE$ changes from $0$ to a non-zero value, then some process $p$ is executing \emph{Try} in its super-passage $sp$ with port $k$, and has not invoked \emph{Promote} (from \emph{Try}) at $t$. Since there are finitely many crashes, $p$ eventually completes such an execution of \emph{Promote} at some time $t'$.  If $LOCK\_STATUS$ changes to $(1,*,*)$ by $t'$, we are done.  Otherwise, we have that $p$ completes an execution of \emph{Promote} in $[t,t']$ and $ACTIVE \neq 0$ throughout this interval.  By~\cref{point in promote}, $LOCK\_STATUS$ changes to $(1, *, *)$ at some time $t_0 \in [t,t']$.
\end{proof}

\StatusChangesFromZeroToOne*
\begin{proof}
Assume there are finitely many crash steps in the execution. Consider the earlier time $t$ when $ACTIVE$ changed to a non-zero value.
If $LOCK\_STATUS=(0,*,*)$ at $t$, then by~\cref{LOCK changes then CS_STATUS changes to 1}, $LOCK\_STATUS$ will change to $(1, *, *)$.
Otherwise, it must be the case that $LOCK\_STATUS$ changes to $(0,*,*)$ at time $t' > t$, while $ACTIVE \neq 0$. Thus, by~\cref{CS_STATUS changed to 0 then changes back}, $LOCK\_STATUS$ will change to $(1, *, *)$.
\end{proof}

\wPortMain*
\begin{proof}
Mutual exclusion follows from~\cref{w-port ME}.  Starvation-freedom follows from~\cref{w-port SF}.
From the code, $p$ returns from \emph{Try} after a finite number of its own non-crash steps once the abort step is signalled,
and so bounded abort is satisfied.
Since $M$ is well-behaved and \emph{Recover} is wait-free, then CS re-entry and wait-free CS re-entry are satisfied.
Since \emph{Exit} is wait-free, wait-free exit is satisfied.  Super-passage wait-free exit is satisfied since \emph{Exit} and \emph{Recover}
are wait-free. A crash-free passage clearly incurs $O(1)$ RMRs (assuming, in the DSM model, that the spin variable is
allocated from local memory). Thus, the algorithm has $O(1)$ passage complexity and $O(1+F)$ super-passage complexity in both CC and DSM models.
The algorithm only uses a constant number of arrays of size $D$ and $O(D^2)$ statically pre-assigned boolean variables, so its space complexity is $O(D^2)$.
\end{proof} 
\section {Proofs Omitted From Section \ref{chap:tree}} \label{appendix:tree proof}

\section {Proofs Omitted From Section \ref{chap:transform}} \label{appendix:transform proof}

\fi

\end{document}